\documentclass[12pt]{IEEEtran}
\usepackage{setspace,amsmath,latexsym,cite,amssymb,epsfig,amsfonts,subfig}
\usepackage{url,cite}
\usepackage{graphicx}
\usepackage{psfrag}
\usepackage{footmisc}
\usepackage{multirow}
\usepackage{color}
\usepackage{soul}
\usepackage{multicol}
\usepackage{mathtools}

\graphicspath{{.}{../eps/}{./images/}{./images/result/}}
\usepackage{amsmath}
\newcommand*\diff{\mathop{}\!\mathrm{d}}

\usepackage{amssymb}
\usepackage{amsthm}

\usepackage[cmintegrals]{newtxmath}
\usepackage{cite}
\usepackage{siunitx}
\usepackage{mathtools}

\usepackage{mathtools, cuted}
\usepackage[font=small, labelsep=period]{caption}
\usepackage{makecell}
\usepackage{tabu}
\usepackage{bbm}

\onecolumn
\doublespace

\newtheorem{lemma}{Lemma}

\DeclareMathOperator{\rank}{rank}

\begin{document}

\title{Securing Visible Light Communication Systems by Beamforming in the Presence of Randomly Distributed Eavesdroppers}

\author{Sunghwan~Cho,~\IEEEmembership{Student Member,~IEEE,}
        Gaojie~Chen,~\IEEEmembership{Member,~IEEE,}
        and~Justin~P.~Coon,~\IEEEmembership{Senior Member,~IEEE}
\thanks{S. Cho, G. Chen and J. P. Coon are with the Department of Engineering Science, University of Oxford, Oxford OX1 3PJ, U.K. (e-mail: \{sunghwan.cho, gaojie.chen, justin.coon\}@eng.ox.ac.uk).}
\thanks{This work was supported by EPSRC grant number EP/N002350/1 (``Spatially Embedded Networks").}}

\maketitle

\begin{abstract}
This paper considers secrecy enhancement mechanisms in visible light communication (VLC) systems with spatially distributed passive eavesdroppers (EDs) under the assumption that there are multiple LED transmitters and one legitimate receiver (UE). Based on certain amplitude constraints, we propose an optimal beamforming scheme to optimize secrecy performance. Contrary to the case where null-steering is made possible by using knowledge of the ED locations, we show that the optimal solution when only statistical information about ED locations is available directs the transmission along a particular eigenmode related to the intensity of the ED process and the intended channel. Then, a sub-optimal LED selection scheme is provided to reduce the secrecy outage probability (SOP). An approximate closed-form for the SOP is derived by using secrecy capacity bounds. All analysis is numerically verified by Monte Carlo simulations. The analysis shows that the optimal beamformer yields superior performance to LED selection. However, LED selection is still a highly efficient suboptimal scheme due to the complexity associated with the use of multiple transmitters in the full beamforming approach. These performance trends and exact relations between system parameters can be used to develop a secure VLC system in the presence of randomly distributed EDs.
\end{abstract}

\IEEEpeerreviewmaketitle

\section{Introduction}
\label{sec:1}
\IEEEPARstart{D}{ue} to the rapid proliferation of mobile communication devices and the associated difficulties in adequately allocating spectra to support new services, visible light communication (VLC) has become an increasingly interesting topic of research in academia and industry. The VLC medium does not interfere with RF systems, and VLC spectrum can be easily reused (spatially) since light can be confined to a certain indoor area. Moreover, VLC uses unregulated spectrum with a wide bandwidth (428 to 750 THz) and is capable of exploiting existing LED light infrastructure for communication~\cite{lifi,5gwillbe}.

Compared to RF channels, VLC exploits light-of-sight (LoS) propagation and has relatively good signal confinement properties.  However, the VLC channel is still of a broadcast nature. Therefore, securing VLC transmissions is an important issue, particularly for deployments in open places such as public libraries, offices, and shopping malls. To cope with the security issue in RF systems, the focus on physical layer security (PLS), which is based on the information theoretic notion of employing coding to achieve secure communication, has accelerated since Wyner's seminal work \cite{A.D75}. Due to the broadcast nature of RF communications, both the legitimate receiver, or user equipment (UE), and eavesdroppers (EDs) may receive data from the source. However, the principle of PLS states that if the capacity of the intended data transmission channel is higher than that of the eavesdropping channel, the data can be transmitted at a rate close to the difference in their capacities, the so-called \emph{secrecy capacity}, so that only the intended receiver can successfully decode the data. 

It is difficult to obtain knowledge of passive ED locations. Yet, the analysis of secrecy capacity in spatial networks inherently depends upon this geometric properties.  The mathematical theory of stochastic geometry is a powerful  tool for dealing with spatial uncertainty~\cite{M.H08,P.C08}. Using stochastic geometric methods, the impact of random ED locations on secrecy performance for RF communications has been investigated in recent years~\cite{X.Z11,G.G14,T.X14,G.C17}. The location distribution of EDs can be modeled as a Poisson point process (PPP) or a binomial point process (BPP). In \cite{X.Z11}, the locations of multiple legitimate pairs and EDs were represented as independent two-dimensional PPPs, and the average secrecy throughput in such a wireless network was studied. Multiple-input multiple-output (MIMO) transmission with beamforming was considered later in \cite{G.G14,T.X14} to enhance  secrecy performance. Transmit antenna selection and full-duplex schemes have also been used to enhance  secrecy performance with randomly located EDs~\cite{G.C17}.

Motivated by the advantage of PLS, a recent topic of interest in the research community has been the investigation of PLS applied in VLC systems using various transmission methods, e.g., beamforming, jamming, etc.  Recently, Lampe et~al analyzed the achievable secrecy rate for  single-input single-output (SISO) and multiple-input single-output (MISO) scenarios and proposed a variety of beamforming schemes such as zero-forcing (null-steering), artificial noise generation, friendly jamming, and robust beamforming~\cite{lampeJSAC, lampeICC,lampeGlobecom}. Additionally, Alouini et~al proposed the truncated normal input distribution and the truncated generalized normal input distribution to increase the secrecy rate under constraints on the input signal amplitude~\cite{alouini1,alouini2}. It is important to note, however, that these contributions assumed a small number of EDs are present in the system and either the channel state information (CSI) or the locations of the EDs are known. In practice, it might be impossible to obtain ED CSI or locations. 

Inspired by the aforementioned contributions exploiting stochastic geometry in RF communications, our previous work~\cite{17S.C} firstly developed an analogous approach to modeling ED locations in VLC systems. In this paper, we use this model to further analyze system performance and propose new MISO beamforming solutions.  The contributions of this paper can be summarized comprehensively as follows:\begin{itemize}
\item we propose a MISO beamforming solution that optimizes secrecy performance measures (e.g., SNR and secrecy capacity bounds) subject to a signal amplitude constraint for VLC systems when only information about the ED intensity measure is available at the transmitter;
\item we demonstrate that the proposed beamforming method is well approximated by a simple LED selection scheme when the distance between the UE and one of the transmitting LEDs is small;
\item we obtain closed-form bounds on the secrecy outage probability (SOP) when LED selection is adopted.
\end{itemize}

The rest of this paper is organized as follows\footnote{The notation and symbols used in the paper are listed in Table~\ref{tb:1}.}. Section II begins with the system model describing the modulation and beamforming schemes in VLC and providing various performance measures. In Section III, the optimal beamformer maximizing secrecy performance is investigated. In Section IV,  LED selection is proposed, and  closed-form  upper and lower bounds on the SOP are calculated. Section V gives numerical results that support out analysis. Section VI concludes the paper. 

\begin{table}[!t] 
\centering
\caption{Notation and Symbols Used in the Paper}
\small
\begin{tabular}{|c|l|}
\hline

   Symbol             &Definition/Explanation         \\ \hline
$L$                      & the length of a room \\
$W$                    & the width of a room \\
$Z$                     & the height from the ceiling to the work plane\\\
$N$                     & number of transmitters \\
$\Phi_{E}$           & poisson point process of EDs \\
$ \lambda_E$      & ED intensity function  \\
$I_{DC}$              & fixed bias current             \\
$R$                      & photodetector's responsivity \\
$\alpha$               & modulation index            \\
$\phi_{1/2}$            & half illuminance angle   \\
$A_{PD}$                & physical area of a photodiode \\
$\phi$                    & angle of irradiance          \\
$\psi$                    & angle of incidence           \\
$\kappa$                       & refractive index of an optical concentrator \\
$\Psi_{c}$             & received field of view of a photodiode  \\

$\mathbb R$         &   set of real numbers       \\
$\mathbb R^+$         &  set of non-negative real numbers      \\
$\mathbbm{1}$         & all-ones column vector \\
$\mathbf{0}$             & all-zeros column vector \\
$\mathbb E[\cdot]$    &   expectation operator       \\
$\mathbb{P}(\cdot)$   &   probability operator        \\
$[\cdot]^T$              & transpose operator \\
$\Gamma(x,y)$         &   upper incomplete gamma function  \\ \hline
\end{tabular}
\label{tb:1}
\end{table}

\section{System Model}
\label{sec:2}

\begin{figure}[t]
  \centering
  \centerline{\includegraphics[scale=0.76]{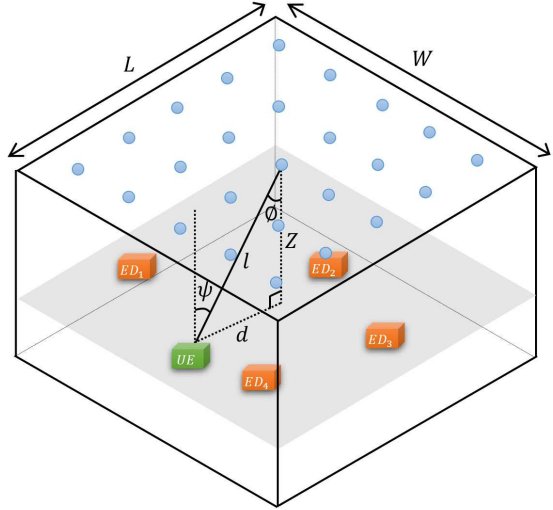}}
  \caption{Rectangular room configuration for VLC systems. $W$ and $L$ are the room's width and length, and $Z$ denotes the height from the ceiling to the work plane. Dots denote LED transmitters.} \label{fig:configuration}
\end{figure}

\subsection{Data Transmission}
We consider the downlink of a VLC system in a rectangular room\footnote{This  may be an open space such as a shopping mall or a large office.} as shown in Fig.~\ref{fig:configuration}, where $W$, $L$, and $Z$ denote the width, the length, and the height of the ceiling relative to the work plane, respectively. We assume that one fixed UE exists and multiple random EDs are randomly distributed according to a PPP $\Phi_{E}$ with  intensity $\lambda_E$  in the room. Note that there is no assumption that the PPP is homogeneous.  We assume all the receiver nodes are located on the same work plane, and $N$ transmitters are attached to the ceiling of the room.  Each transmitter --- i.e., an LED fixture consisting of multiple individual LEDs --- is assumed to be capable of communicating independently of other transmitters~\cite{lampeJSAC}.  We assume that EDs act independently of one another (i.e., there is no collusion).

A DC-biased pulse-amplitude modulation (PAM) VLC scheme is considered~\cite{lampeJSAC, lampeICC}. The data signal $x(t) \in \mathbb{R}$ in time slot $t$ is a zero-mean current signal superimposed on a fixed bias current $I_{DC} \in \mathbb{R}^+$. The fixed bias $I_{DC}$ is used for the  purpose of illumination. To maintain  linear current-to-light conversion, the amplitude of the modulated signal $x(t)$ is constrained such that $ | x(t) |  \leq \alpha I_{DC}$, where $\alpha \in [0,1]$ is termed  the modulation index. Thus, the dynamic range of the LED is $I_{DC} \pm \alpha I_{DC}$. Also, since $\mathbb{E} [x(t) ]=0$, the data signal does not affect  illumination.

The VLC channel model can be written as
\begin{equation}
y(t) = h x(t) + n(t)
\end{equation}
where $h$ is the channel transfer coefficient and $n(t)$ is the zero-mean additive white Gaussian noise (AWGN) at a receiver. According to~\cite{TKomine}, the channel gain $h$ in a VLC system corresponding to an LED with a generalized Lambertian emission pattern is given by
\begin{equation} \label{channelgain}
h=\left\{ \begin{array}{ccc}
 \eta\displaystyle\frac{(m+1) A_{PD}} {2\pi l^2}  \frac{\kappa^2 \cos^m(\phi) }{\sin^2(\Psi_{c}) }\cos(\psi) R T & \mbox{for} & |\psi| \leq  \Psi_{c}, \\
 0  & \mbox{for} & |\psi|>\Psi_{c}
\end{array}\right.
\end{equation}
where $\eta$ ($W/A$) is the current-to-light conversion efficiency and $m=-\ln(2)/\ln(\cos(\phi_{1/2}))$ is the order of Lambertian emission with half illuminance at $\phi_{1/2}$, and $A_{PD}$ is the physical area of the photodiode (PD). As shown in Fig.~\ref{fig:configuration}, $l$ is the distance between the transmitter and the receiver, and $d$ denotes the distance between the transmitter and the receiver in the work plane. $\phi$ is the angle of irradiance, and $\psi$ is the angle of incidence. Also, $\kappa$ is the refractive index of the optical concentrator at the receiver, $\Psi_{c}$ denotes the received field of view of the PD, $R$ is the photodetector's responsivity, and $T$ ($V/A$) is the transimpedance amplifier gain. Note that this channel model considers only an LoS component. Moreover, by assuming that a receiver's PD faces up normal to the work plane, we can rewrite (\ref{channelgain}) in terms of $l$ as
\begin{equation} \label{simplified_h}
\begin{split}
h = & \eta \displaystyle\frac{(m+1) A_{PD}} {2\pi l^2}   \frac{\kappa^2}{\sin^2(\Psi_{c}) } \left(\frac{Z}{l}\right)^m \left(\frac{Z}{l}\right) RT=Kl^{-(m+3)}
\end{split}
\end{equation}
where $K=\left( \eta (m+1)A_{PD} Z^{m+1} \kappa^2 RT \right)/\left(2 \pi \sin^2(\Psi_c)\right)$.

As in~\cite{lampeJSAC}, we define a beamforming vector $\mathbf{w} = [ w_1, w_2, ..., w_N ]^T$, where $w_i$ for $i \in \{1,2,...,N \}$ is a weight for the $i$th transmitter and $|w_i| \leq 1$. Thus, the transmitted signal vector $\mathbf{x}(t) \in \mathbb{R}^N$ can be written as $\mathbf{x} (t) = \mathbf{w} s(t)$, where $s(t)$ is the transmitted data symbol. Accordingly, the transmitted signal $\mathbf{x}(t)$ is subject to the amplitude constraint $|\mathbf{x}(t)| \preceq \alpha I_{DC} \mathbbm{1}$. Therefore, the received signal at the UE and eavesdropper $E_e$ with $e \in \Phi_{E}$ can be described as
\begin{subequations}
\begin{align}
y_U(t) &=  \mathbf{h_U^T} \mathbf{x}(t) + n_U(t), \\
y_{E_{e}}(t) &= \mathbf{h_{E_{e}}^T}\mathbf{x}(t) + n_{E_{e}}(t)
\end{align}
\end{subequations}
respectively, where $\mathbf{h_U}$ and $\mathbf{h_{E_{e}}} \in \mathbb{R}^N$ are the channel gain vectors from the transmitters to the UE and eavesdropper $E_e$, respectively, and $n_U$ and $n_{E_e}$ are zero-mean AWGN random variables at the UE and eavesdropper $E_e$, each with variance $\sigma^2$. For notational convenience, the time index $t$ is ignored for the remainder of the paper.

\subsection{Performance Measures}
For Gaussian VLC channels with  amplitude constraints, the peak signal-to-noise ratio (SNR) at the UE and the eavesdropper $E_e$ can be written as
\begin{subequations}
\begin{align} \label{eq:SNR1}
\gamma_U &= \displaystyle \frac{\alpha^2 I_{DC}^2 \mathbf{w^T}\mathbf{h_U}\mathbf{h_U^T}\mathbf{w}}{\sigma^2}, \\ \label{eq:SNR2}
\gamma_{E_e} &= \displaystyle \frac{\alpha^2 I_{DC}^2 \mathbf{w^T}\mathbf{h_{E_e}}\mathbf{h_{E_e}^T}\mathbf{w}}{\sigma^2}.
\end{align}
\end{subequations}
The capacity of the  VLC channel is given by~\cite{optic_capacity}
\begin{equation} \label{eq:definition_capacity_vlc}
C = \max_{p_X} \mathbb{I}(X;Y) 
\end{equation}
where $p_X$ is the input distribution and $\mathbb{I}(\cdot;\cdot)$ denotes the mutual information. Note that the random variable $X$ has an amplitude constraint, i.e., $| X | \leq \alpha I_{DC}$. It is infeasible to calculate the closed-form solution for (\ref{eq:definition_capacity_vlc}) due to this amplitude constraint~\cite{capacity_closedform}. Thus, the capacity upper and lower bounds are used for our analysis, which are given in~\cite[Theorem 5]{optic_capacity} as
\begin{subequations}
\label{eq:definition_capacity1}
\begin{align} \label{eq:definition_capacity}
&C^{\text{upper}} =\displaystyle \frac{1}{2} \log\left(1+\gamma\right), \\ \label{eq:definition_capacity2}
&C^{\text{lower}}= \displaystyle\frac{1}{2} \log\left(1+\frac{2 \gamma}{\pi e}\right)
\end{align}
\end{subequations}
where $\gamma$ is the received SNR.

In addition to that, we define the SOP as the probability that the secrecy capacity $C_s$ is lower than a threshold secrecy rate $C_{\text{th}}$, i.e.,
\begin{equation}
P_{\text{SO}}= \mathbb{P} ( C_s  \leq C_{\text{th}}).
\end{equation}
However, since the closed-form of the secrecy capacity with the input amplitude constraint is also not readily available, we employ the lower and upper bounds on secrecy capacity as defined in~\cite[Theorem 1]{lampeJSAC}, which are given by
\begin{subequations}
\begin{align}
C_s^{\text{lower}}&=\frac{1}{2} \log\left(\frac{6 \gamma_U + 3 \pi e}{\pi e \gamma_E^* + 3\pi e}\right), \\
C_s^{\text{upper}}&=\frac{1}{2} \log\left(\frac{\gamma_U + 1}{\gamma_E^* + 1}\right)
\end{align}
\end{subequations}
where $\gamma_E^*=\underset{e \in \Phi_E}{\max}~{\gamma_{E_e}}$ (i.e., the \emph{worst case} ED with the highest SNR). Applying these bounds yields the following upper and lower bounds on the SOP:
\begin{subequations}\label{SOP}
\begin{align}
P_{\text{SO}}^{\text{upper}}&= \mathbb{P} ( C_s^{\text{lower}}  \leq C_{\text{th}}), \\ \label{SOP_upper}
P_{\text{SO}}^{\text{lower}}&= \mathbb{P} ( C_s^{\text{upper}}  \leq C_{\text{th}}).
\end{align}
\end{subequations}

\section{Optimal Beamforming}
\label{sec:3}
In this section, we propose beamformer designs based on the formulation of several optimisation problems that aim to improve secrecy performance when only information about the intensity of the ED PPP is known.  Crucially, we demonstrate that the proposed beamforming solutions apply to both homogenous and inhomogeneous ED processes.

\subsection{Optimization Based on SNR}
Without knowledge of ED locations, a natural objective is to minimize the average SNR of EDs $\overline{\gamma}_E$ subject to a constraint on the minimum require UE SNR  $ {\gamma}_U$, same as in RF communications~\cite{09I.K, 15H.H ,16G.C}. A related, alternative objective may be to maximize ${\gamma}_U$ subject to a constraint on $\overline{\gamma}_E$. In this subsection, both of these cases will be investigated.

\subsubsection{Minimizing Average Eavesdropper  SNR} \label{sec:MES}
The SNR of the UE~(\ref{eq:SNR1}) can be written as
\begin{equation} \label{eq:SNR_U}
\gamma_U = \varphi \mathbf{w^T A w}
\end{equation}
where $\varphi=\alpha^2 I_{DC}^2 / \sigma^2$ and $\mathbf{A}=\mathbf{h_U h_U^T}$. Note that the rank of $\mathbf{A}$ is one.
Also, from (\ref{eq:SNR2}), the average SNR of an ED  can be written as
\begin{equation} \label{average_E_SNR}
\overline{\gamma}_E  = \mathbb{E} [\varphi \mathbf{w^T h_{E_e} h_{E_e}^T w} ] =\varphi \mathbf{w^T} \mathbb{E} [\mathbf{h_{E_e} h_{E_e}^T}] \mathbf{w} =\varphi \mathbf{w^T} \overline{\mathbf{B}} \mathbf{w}
\end{equation}
where $\overline{\mathbf{B}}=\mathbb{E} [\mathbf{h_{E_e} h_{E_e}^T}]$. The element in the $i$th row and $j$th column of $\overline{\mathbf{B}}$ is given by
\begin{equation} \label{eq:B}
\overline{B}_{i,j} = \frac{1}{N_E}\int_{\frac{-L}{2}} ^{\frac{L}{2}} \int_{\frac{-W}{2}}^{\frac{W}{2}} \frac{{\lambda_E(x, y)} K^2 }{l_i^{m+3}(x,y) l_j^{m+3}(x,y)}  \diff x \diff y\\
\end{equation}
where $\lambda_E(x,y)$ is the intensity of EDs at the point $(x, y)$ and $l_i(x,y)$ for $i \in \{1,2,\cdots,N\}$ is the distance between the $i$th transmitter and the point $(x, y)$. Note that $\lambda_E (x,y)$ is a constant when the ED point process is homogeneous. Also, $N_E$ denotes the average number of EDs, which is given by
\begin{equation}
N_E = \int_{\frac{-L}{2}} ^{\frac{L}{2}} \int_{\frac{-W}{2}}^{\frac{W}{2}} {\lambda_E(x, y)}  \diff x \diff y.
\end{equation}

From the formulation given above, it is clear that the optimal beamforming vector $\mathbf{w}^*$ is given by
\begin{subequations}\label{case1}
\begin{align}\label{objective1}
&\mathbf{w^*}=\arg\underset{\mathbf{w}}{\min}~ \varphi \mathbf{w^T \overline{B} w} \\ \label{eq:constraint1}
&\text{s.\,t.}\, \left\{ \begin{array}{l}
\varphi \mathbf{w^T A w} \geq \rho_U \\
|\mathbf{w}| \preceq \mathbbm{1}
\end{array} \right.
\end{align}
\end{subequations}
where $\rho_U$ denotes the required SNR of the UE. Note that the optimal beamformer $\mathbf{w}^*$ is $\mathbf{0} \in \mathbb{R}^N$ without the first constraint. Then, we form the Lagrangian~\cite{04S.B} as follows
\begin{equation}\label{eq:Lagrange1}
\mathcal{L} = \varphi \mathbf{w^T \overline{B} w} - \lambda \left(\varphi \mathbf{w^T A w} - \rho_U \right) - \mathbf{\mu_-^T}(\mathbf{w}+\mathbbm{1})+\mathbf{\mu_+^T}(\mathbf{w}-\mathbbm{1})
\end{equation}
where $\lambda\in \mathbb{R}$ and $\mathbf{\mu}_-$, $\mathbf{\mu}_+ \in \mathbb{R}^N$ are the Lagrange multipliers. To let $\mathcal{L}$ have the non-trivial minimum value with respect to $\mathbf{w}$ and analytically calculate the optimal solution of (\ref{case1}), the condition has to be satisfied as 
\begin{equation}
\mathbf{\mu}_-=\mathbf{\mu}_+=\mathbf{0}.
\label{eq:slack}
\end{equation}
According to the local sensitivity analysis in~\cite{04S.B}, zero Lagrange multipliers $\mathbf{\mu}_-$ and $\mathbf{\mu}_+$ imply that the second inequality constraint is slack, i.e., $|w_i| \neq 1$ for $i \in \{1, 2, ..., N\}$. 

If (\ref{eq:slack}) can be satisfied\footnote{If $\rho_U$ can be achieved by only the transmission of the nearest LED to the UE, the beamformers of other transmitters must not be $\pm1$. Moreover, without loss of generality, we can assume that the beamforming element of the nearest LED can be very close to $\pm 1$, but not equals to $\pm 1$. Then, (\ref{eq:slack}) can be satisfied.}, computing the partial derivative of $\mathcal{L}$ with respect to $\mathbf{w}$ and setting the result equal to zero leads to 
\begin{equation} \label{eq:Lagrange_equality1}
(\mathbf{\overline{B}}-\lambda \mathbf{A}) \mathbf{w} = \mathbf{0}.
\end{equation}
The optimal beamforming vector must satisfy (\ref{eq:Lagrange_equality1}). If $\lambda=0$ (i.e., the UE SNR constraint is inactive), $\mathbf{w}$ should belong to the null space of $\mathbf{\overline{B}}$.  Referring to (\ref{eq:B}), it should be clear that the rank of $\mathbf{\overline{B}}$ depends on the intensity of the ED process.  Indeed, it is possible that the intensity is such that $\mathbf{\overline{B}}$ is reduced rank. On the other hand, if $\mathbf{\overline{B}}$ is full rank\footnote{This condition can be confirmed when the ED PPP is homogeneous, for example.}, we have
\begin{equation} \label{eq:Lagrange_equality2}
\overline{\mathbf{B}}^{-1}\mathbf{A} \mathbf{w}= \eta  \mathbf{w}
\end{equation}
where $\eta=1/\lambda$. This implies that $\eta$ is the eigenvalue of $\overline{\mathbf{B}}^{-1}\mathbf{A}$ and $\mathbf{w}$ is the corresponding eigenvector.  Hence, the optimal solution satisfies
\begin{equation} \label{eq:Lagrange}
\overline{\gamma}_E = \varphi \mathbf{w^T\overline{B}w}= \varphi \frac{1}{\eta} \mathbf{w^TAw} \geq  \frac{1}{\eta_{\text{max}}}\rho_U.
\end{equation}
From this, we deduce that the minimum $\overline{\gamma}_E$ is  inversely proportional to the maximum eigenvalue $\eta_{\text{max}}$. Here, since $\mathbf{A}$ is rank one, we have that $\mathbf{\overline{B}^{-1}A}$ is rank one since
\begin{equation}
  0 < \rank\!\left(\mathbf{\overline{B}^{-1}A}\right) \leq \min\!\left\{\rank\!\left(\mathbf{\overline{B}}\right),\rank(\mathbf{A})\right\} = 1.
\end{equation}
Hence, there  exists a single non-zero eigenvalue $\eta_{\text{max}}$, and the optimal beamformer $\mathbf{w^*}$ is obtained by scaling the corresponding eigenvector such that $\varphi \mathbf{{w^*}^TA{w^*}} = \rho_U$ and $|\mathbf{w^*}| \prec \mathbbm{1}$. When feasible, the fact that only one non-zero eigenvalue exists implies the solution is unique.  Hence, it is clear that this approach  gives a simple method of calculating the beamforming vector. 

On the other hand, if $\rho_U$ is high so that (\ref{eq:slack}) cannot be satisfied, the optimal solution $\mathbf{w^*}$ can be calculated numerically, since the objective function (\ref{objective1}) is convex (Note $\overline{\mathbf{B}}$ is a positive semidefinite matrix).

\subsubsection{Maximizing User SNR} \label{sec:MUS}
We now investigate the problem of maximizing the SNR of the UE $\gamma_U$ while constraining  the average SNR of the EDs $\overline{\gamma}_E$.  Similar to the previous subsection, we can formulate the optimization problem as
\begin{subequations} \label{eq:optimization_problem2}
\begin{align} \label{eq:optimization_problem2a}
&\mathbf{w^*}=\arg\underset{\mathbf{w}}{\max}~ \varphi \mathbf{w^T A w} \\ \label{eq:constraint2}
&\text{s.\,t.}\, \left\{ \begin{array}{l}
\varphi \mathbf{w^T \overline{B} w} \leq \overline{\rho}_E \\
|\mathbf{w}| \preceq \mathbbm{1}
\end{array} \right.
\end{align}
\end{subequations}
where $\overline{\rho}_E$ is the target constraint on the average SNR of EDs. 

Following the same approach as above leads to an analogous result when $\mathbf{\overline{B}}$ is nonsingular:
\begin{equation}
\mathbf{\overline{B}^{-1} A w}= \lambda \mathbf{w}.
\end{equation}
Hence, $\lambda$ and $\mathbf{w}$ are an eigenvalue-eigenvector pair for $\mathbf{\overline{B}^{-1} A}$. The optimal solution satisfies
\begin{equation} \label{eq:Lagrange3}
{\gamma}_U = \varphi \mathbf{w^TAw}= \varphi \lambda \mathbf{w^T \overline{B}w} \leq \lambda_{\text{max}} \overline{\rho}_E
\end{equation}
where $\lambda_{\text{max}}$ is the maximum (indeed, the only nonzero) eigenvalue, and the optimal beamformer $\mathbf{w}^*$ is the associated eigenvector after being scaled such that $ \varphi \mathbf{{w^*}^T\overline{B}{w^*}} =  \overline{\rho}_E$ and $|\mathbf{w}| \prec \mathbbm{1}$ (if feasible).

\subsection{Optimization Based on Capacity}
In this subsection, we will investigate the optimal beamformer based on the capacities of the UE and EDs.

\subsubsection{Minimizing Eavesdropper Average Capacity}
Since the closed-form expression for the capacity of a VLC channel is not available, we will use the upper and lower bounds given in~(\ref{eq:definition_capacity1}). Thus, the optimal beamformer $\mathbf{w}^*$ minimizing the average ED capacity upper bound is given by
\begin{subequations}\label{eq:beamformer_C1}
\begin{eqnarray} \label{eq:beamformer_C1a}
&&\mathbf{w^*}= \arg\underset{\mathbf{w}}{\min}~\mathbb{E} \left[ \displaystyle\frac{1}{2} \log\left(1+\gamma_{E_e} \right) \right] \\
&&\text{s.\,t.}\, \left\{ \begin{array}{l}
\frac{1}{2}\log(1+\frac{2 \gamma_U}{\pi e}) \geq \xi_U \\
|\mathbf{w}| \preceq \mathbbm{1}
\end{array} \right.
\end{eqnarray}
\end{subequations}
where $\xi_U$ is the required capacity lower bound for the UE. However, since the objective function (\ref{eq:beamformer_C1a}) does not lend itself to tractable analysis and a practically implementable solution, we revise the objective function using Jensen's inequality, i.e., it becomes $(1/2) \log(1+ \overline{\gamma}_E)$, which allows us to obtain the suboptimal beamformer $\mathbf{w^*}$. In addition, since $\log(\cdot)$ is a monotonic increasing function, (\ref{eq:beamformer_C1}) can be reformulated to
\begin{subequations} \label{case3}
\begin{eqnarray} \label{eq:beamformer_C1_modified}
&&\mathbf{w^*}= \arg\underset{\mathbf{w}}{\min}~\varphi \mathbf{w^T \overline{B} w}  \\
&&\text{s.\,t.}\, \left\{ \begin{array}{l}
\varphi \mathbf{w^T A w} \geq M_1 \\
|\mathbf{w}| \preceq \mathbbm{1}
\end{array} \right.
\end{eqnarray}
\end{subequations}
where $M_1=\left(2^{2\xi_U}-1\right){\pi e}/{2}$. Note that (\ref{case3}) has an identical form to (\ref{case1}). It follows that the suboptimal beamformer $\mathbf{w^*}$ is the eigenvector corresponding to the maximum eigenvalue of $\mathbf{\overline{B}^{-1} A}$ after being scaled such that $\varphi \mathbf{{w^*}^TA{w^*}} = M_1$ and $|\mathbf{w}| \prec \mathbbm{1}$.

\subsubsection{Maximizing User Capacity}
The optimization problem for maximizing the lower bound of the UE's capacity subject to an ED capacity constraint  can be formulated as
\begin{subequations} \label{eq:beamformer_C2}
\begin{eqnarray}
&&\mathbf{w^*}=\arg \underset{\mathbf{w}}{\max} \displaystyle\frac{1}{2} \log\left(1+\frac{2 \gamma_{U}}{\pi e} \right)  \\
&&\text{s.\,t.}\, \left\{ \begin{array}{l}
\frac{1}{2}  \mathbb{E}\left[\log(1+ \gamma_{E_e})\right] \leq \overline{\xi}_E \\ \label{eq:cost_function1}
|\mathbf{w}| \preceq \mathbbm{1}
\end{array} \right.
\end{eqnarray}
\end{subequations}
where $\overline{\xi}_E$ is the target constraint for the average capacity upper bound of EDs. Again, by applying Jensen's inequality to the constraint,  we arrive at the alternative formulation
\begin{subequations} \label{eq:beamformer_C2_modified}
\begin{eqnarray}
&&\mathbf{w^*}=\arg \underset{\mathbf{w}}{\max}~\mathbf{\varphi w^TAw}  \\
&&\text{s.\,t.}\, \left\{ \begin{array}{l}
\varphi \mathbf{w^T \overline{B} w} \leq M_2 \\ \label{eq:cost_function2}
|\mathbf{w}| \preceq \mathbbm{1}
\end{array} \right.
\end{eqnarray}
\end{subequations}
where $M_2=2^{2{\overline{\xi}}_E}-1$. Hence, we deduce that the suboptimal beamforming vector is the eigenvector corresponding to the maximum eigenvalue of $\mathbf{\overline{B}^{-1} A}$ scaled appropriately.

\begin{figure}
\centering\subfloat[The average SNR of EDs for a UE at $(0,1)$]{ \label{fig:Fig2a}\includegraphics[width=0.43\textwidth]{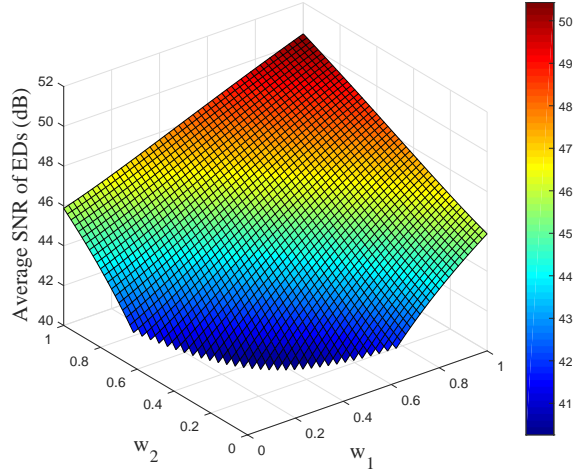}}    \\       
\centering\subfloat[The average SNR of EDs for a UE at $(2,1)$]{ \label{fig:Fig2b}\includegraphics[width=0.43\textwidth]{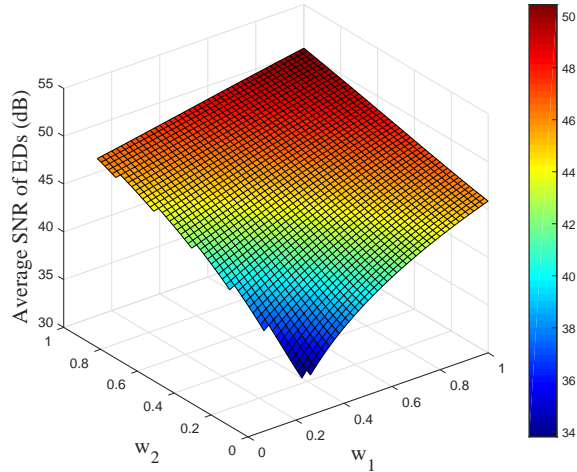}}            
 \caption{The average SNR of EDs for  different UE locations plotted against beamforming weights. Two transmitters $T_1$ and $T_2$ are located at $(2.5, 0)$ and $(-2.5,0)$, respectively. The room size is $L=8~\mathrm{m}$ and $W=8~\mathrm{m}$. ${\rho}_U=40~\mathrm{dB}$ is applied.}  \label{fig:Fig2}
\end{figure}

\subsection{Comparing MISO Beamforming to LED Selection}
In the previous subsection, we showed that the optimal beamformer for SNR and rate objectives is universally related to the maximum eigenmode of $\mathbf{\overline{B}^{-1}A}$.  The proposed optimal beamforming vector cannot be a null-steering solution as was the case in~\cite{lampeICC} unless a (perhaps pathological) condition occurs that makes $\mathbf{\overline{B}}$ singular.  

Since $\mathbf{A}$ depends on the UE's location and $\mathbf{\overline{B}}$ depends on the transmitter locations and the intensity function of the EDs, we can note that when the UE is located near to a transmitter, the optimal beamformer looks like a transmitter selection process, i.e., the nearest transmitter's weight is dominant to the others. This is because the eigenvector corresponding to the maximum eigenvalue is significantly affected by the maximum diagonal element of $\mathbf{A}$ when the UE is near to the transmitter.

To illustrate this point, let us take an example of a simple VLC network where two transmitters are located at $T_1=(2.5,0)$ and $T_2=(-2.5,0)$ in an $8 \times 8~\mathrm{m}^2$ square room. The  center of the room is located at the origin of our coordinate system. Figs.~\ref{fig:Fig2}(a) and (b) show the average SNR of EDs according to a different set of beamformer weights when the UE is located at $(0,1)$ and $(2,1)$, respectively. The required SNR of the UE is $\rho_U=40~\mathrm{dB}$. In Fig.~\ref{fig:Fig2}(a), since the UE is located at the exact middle point of the two transmitters, the weight values for the two transmitters that minimize the average SNR of the EDs are equivalent, i.e., $\mathbf{w^*} \approx (0.3, 0.3)$. However, when the UE is nearer to $T_1$ as in Fig.~\ref{fig:Fig2}(b), we can see that the optimal beamformer resembles  LED selection.  More specifically, $\mathbf{w^*} \approx (0.24, 0.02)$ in this example.  We can thus surmise that secrecy performance will be similar for  optimal beamforming and  transmitter selection when the UE is  located close to a transmitter\footnote{We will further explore the secrecy performance for both schemes later in Section V.}.  

With regard to practicalities of implementation, the complexity of the beamforming scheme can be reasonably high due to the use of multiple transmitters. Even though we can efficiently find an eigenvector related to the maximum eigenvalue by using the power method or the Rayleigh quotient method, the computational complexity still can be significant since complexity grows with $N^2$. Additionally, it might not be practical to accurately estimate the intensity function that describes  ED locations.  These arguments  motivate further investigation of the performance of LED selection in the context of VLC systems with randomly distributed EDs.

\section{LED Selection}
\label{sec:4}

\begin{figure}[t]
\centering
\centerline{\includegraphics[scale=0.45]{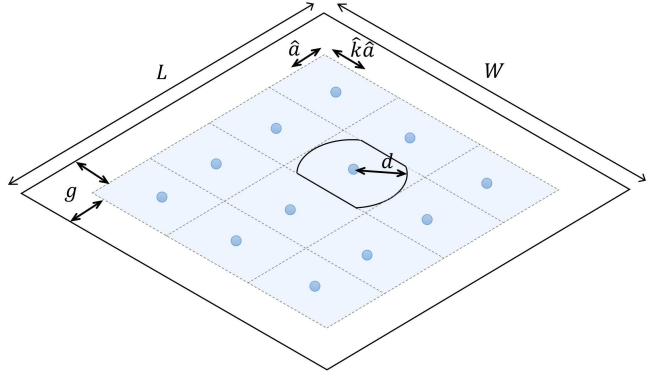}}
\caption{The room configuration for LED selection.}
\label{fig:SOP_configuration}
\end{figure}

We turn our attention to the simple, but suboptimal method of LED selection. In the LED selection scheme, the nearest transmitter to the UE is  selected to transmit the information bearing signal. In this section, we first investigate the SNR and capacity performance metrics with  LED selection. We then analyze the SOP. Closed-form expressions for the upper and lower bounds of the SOP with the LED selection are derived. 

The room configuration in Fig.~\ref{fig:SOP_configuration} is used for our analysis, where multiple transmitters are attached to the ceiling such that the coverage area is identical (but translated in the work plane) for each of them. The parameters $\hat{a}$ and $\hat{k}$ in the figure denote the half length of the rectangular coverage area's width and the ratio of length to width of the coverage area.  We assume that the UE can be located only within  the coverage area\footnote{One may think of this restriction as a policy instigated to guarantee the security of UE, through the restriction of the  UE to a ``safety zone''.  On the other hand, it is fairly easy to see that this model reflects many realistic scenarios.}, i.e., the shaded area, while multiple EDs can (randomly) position themselves anywhere in the entire area of the room.  By defining the number of rows and columns of the LED arrangement as $N_r$ and $N_c$, respectively, the relations $\hat{a}=(L/2-g)/N_r$ and $\hat{k}=(W/2-g)/(N_c \hat{a})$ can be deduced, where $g$ denotes the thickness of the edge zone.  The UE and the EDs are assumed to be uniformly distributed according to homogeneous BPP (with one point) and a homogeneous PPP $\Phi_E$ with intensity  $\lambda_E$, respectively.

\subsection{SNR and Capacity Analysis}
Since only the nearest transmitter is selected to transmit a signal to the UE, the vector representing LED selection $\mathbf{w}_s$ can be described by
\begin{equation}
\mathbf{w}_s = \omega\,\mathbf{e}_i
\end{equation}
where $\omega$ is the weight of the selected transmitter and $\mathbf{e}_i$ is the $i$th column of the identity matrix corresponding to the maximum diagonal element of $\mathbf{A}$, i.e., the $i$th transmitter is the nearest. Therefore, the average SNR of the EDs when  LED selection is employed is given by
\begin{equation} \label{eq:SNR_E_selection}
\overline{\gamma}_E = \varphi \mathbf{w_s^T\overline{B}w_s}.
\end{equation}
If there exists a required SNR for the UE ($\rho_U$, as in (\ref{case1})), $\omega$ should satisfy $\omega^2 {A}_{i,i}^2 \geq \rho_U$, where $A_{i,i}$ denotes the maximum diagonal element of $\mathbf{A}$. Similarly, the SNR of the UE can be described by
\begin{equation} \label{eq:SNR_U_selection}
\gamma_U = \varphi \mathbf{w_s^TAw_s}.
\end{equation}
Also, if there exists a constraint on the average SNR of the EDs ($\overline{\rho}_E$, as in (\ref{eq:optimization_problem2})), $\omega$ should satisfy $\omega^2 \overline{B}_{i,i}^2 \leq \overline{\rho}_E$, where  $\overline{B}_{i,i}$ is the $i$th diagonal element of $\mathbf{\overline{B}}$.

Moreover, the average upper bound on the capacity of the ED channel with  LED selection can be written as
\begin{equation}
\overline{C}_{E}^{\text{upper}}= \mathbb{E} \left[ \displaystyle\frac{1}{2} \log\left(1+\varphi \omega^2 h_{i,E_e}^2 \right) \right]
\end{equation}
where $h_{i,E_e}$ is the channel gain from the $i$th (i.e., the optimal) transmitter to eavesdropper $E_e$.
Also, the lower bound on the capacity of the UE is given by
\begin{equation}
{C}_{U}^{\text{lower}}= \displaystyle\frac{1}{2} \log\left(1+\frac{2\varphi \mathbf{w_s^TAw_s}}{\pi e} \right).
\end{equation}

\subsection{Secrecy Outage Probability}
Here, we calculate upper and lower bounds on the SOP for the LED selection scheme.  To this end, we necessarily must know something about the SNR statistics for the UE and the \emph{nearest} ED, hence, the \emph{worst case} ED.  Thus, we begin by providing results for the  probability density function (PDF) and the cumulative distribution function (CDF) for these random variables. To simplify the calculation of the SOP, we assume the weight of the selected transmitter's beamforming vector is always one.  This assumption is justified since, as we will see, the beamforming weight itself does not significantly affect the SOP when LED selection is adopted.

\begin{lemma}
The CDF and PDF of the received UE SNR  $\gamma_U$ are given by (\ref{eq:CDF_SNR_U}) and (\ref{eq:PDF_SNR_U}) at the top of this page, respectively, where $\zeta=(\alpha^2 I_{DC}^2 K^2)/\sigma^2$, and $y_i$ and $K_i$ for $i \in \{1,2,3,4\}$ are given by
\begin{figure*}[t]
\begin{equation} \begin{split}\label{eq:CDF_SNR_U}
F_{\gamma_U}(y) = \left \{ \begin{array}{lll}
  \displaystyle 1-\left({K_4\left(\left( y/\zeta \right)^{\frac{-1}{m+3}}-Z^2\right)+K_3\left(\left( y/\zeta \right)^{\frac{-1}{m+3}}-Z^2\right)^{3/2}}\right)/({4 \hat{k}{\hat{a}}^2})     &\mbox{for}& y_1 < y \leq y_2 \\ \\
  \displaystyle 1-\left({K_2\left(\left( y/\zeta \right)^{\frac{-1}{m+3}}-Z^2\right)+K_1\left(\left( y/\zeta \right)^{\frac{-1}{m+3}}-Z^2\right)^{3/2}}\right)/({4 \hat{k}\hat{a}^2})  &\mbox{for} &y_2 < y \leq y_3  \\ \\
  \displaystyle 1-{\pi\left(\left( y/\zeta \right)^{\frac{-1}{m+3}}-Z^2\right) }/({4 \hat{k}\hat{a}^2})  &\mbox{for} &y_3 < y \leq y_4
\end{array} \right.  \\
\end{split}
\end{equation}
\hrulefill
\end{figure*}
\begin{figure*}[t]
\begin{equation} \begin{split}\label{eq:PDF_SNR_U}
f_{\gamma_U}(y)= \left \{ \begin{array}{lll}
{\left(y/\zeta\right)^{-\frac{1}{m+3}} \left(3  {K_3}\sqrt{\left(y/\zeta\right)^{-\frac{1}{m+3}}-Z^2}+2  {K_4}\right)} /({8 {\hat{a}}^2  {\hat{k}} (m+3) y})  &\mbox{for}& y_1 < y \leq y_2 \\ \\
{\left(y/\zeta\right)^{-\frac{1}{m+3}} \left(3  {K_1}\sqrt{\left(y/\zeta\right)^{-\frac{1}{m+3}}-Z^2}+2  {K_2}\right)} /({8 {\hat{a}}^2  {\hat{k}} (m+3) y}) &\mbox{for} &y_2 < y \leq y_3  \\ \\
 {\pi  \left(y/\zeta\right)^{-\frac{1}{m+3}}} /({4  {\hat{a}}^2 {\hat{k}} (m+3) y}) &\mbox{for} &y_3 < y \leq y_4
\end{array} \right.
\end{split}
\end{equation}
\hrulefill
\end{figure*}
\begin{equation}\begin{split}
y_1&=\zeta(\hat{a}^2(1+\hat{k}^2)+Z^2)^{-3-m},  \\
y_2&=\zeta(\hat{a}^2\hat{k}^2+Z^2)^{-3-m},  \\
y_3&=\zeta(\hat{a}^2+Z^2)^{-3-m},  \\
y_4&=\zeta Z^{-2(3+m)}
\end{split}
\end{equation}
and
\begin{equation}
\begin{split}\label{eq:K}
K_1&=\frac{\left( {2 \sqrt{\hat{k}^2 -1}}/{\hat{k}^2} -2\arccos{\left(\frac{1}{\hat{k}}\right)  }\right)}{a(\hat{k}-1)},  \\
K_2&=\pi - \hat{a} K_1,\\
K_3&=2 \left(\arccos{\left(\frac{1}{\hat{k}}\right)} -\arccos{\left(\frac{1}{\sqrt{\hat{k}^2+1}}\right)}-\arccos{\left(\frac{\hat{k}}{\sqrt{\hat{k}^2+1}}\right)} \right.  \\
     &~ \left. +\frac{2\hat{k}}{\hat{k}^2+1}-\frac{\sqrt{\hat{k}^2-1}}{\hat{k}^2}\right)/{\left(\hat{a}(\sqrt{\hat{k}^2+1} -\hat{k}) \right)},  \\
K_4&=\pi-2 \left( \arccos{\left(\frac{1}{\hat{k}} \right)} - \frac{\sqrt{\hat{k}^2-1}}{\hat{k}^2}  \right) - \hat{k} \hat{a} K_3.
\end{split}
\end{equation}
\end{lemma}

\begin{proof}
See Appendix \ref{sec:app_A}.
\end{proof}

\begin{lemma}
The CDF and PDF of the received SNR for the nearest ED relative to the selected transmitter $\gamma_E^*$  are given by
\begin{subequations}
\begin{align}
F_{\gamma_E^*}(x)&= e^{\lambda_E \pi \left( Z^2 -\left(\frac{x}{\zeta} \right)^{-\frac{1}{m+3}}\right)}, \\
f_{\gamma_E^*}(x)&=\frac{\lambda_E \pi \left(\frac{x}{\zeta} \right)^{-\frac{1}{m+3}}}{x(m+3)} e^{\lambda_E \pi \left( \left(\frac{x}{\zeta} \right)^{-\frac{1}{m+3}}-Z^2\right)}
\end{align}\end{subequations}
for $0 \leq x \leq \zeta Z^{-2(m+3)}$.
\end{lemma}
\begin{proof}
See Appendix \ref{sec:app_B}.
\end{proof}

According to (\ref{SOP}), the upper and lower bounds of the SOP can be written as
\begin{align} 
P_{\text{SO}}^{\text{upper}}&= \mathbb{P} ( C_s^{\text{lower}}  \leq C_{\text{th}}) \nonumber\\
&= \mathbb{P}\left(\frac{6\gamma_U+3\pi e}{\pi e \gamma^*_{E}+3 \pi e} \leq 2^{2C_{\text{th}}}\right)  \nonumber\\
&=\mathbb{P}(\gamma_U \leq a \gamma^*_{E} + 3a -\pi e/2 )
\end{align}
and 
\begin{align}
P_{\text{SO}}^{\text{lower}}&= \mathbb{P} ( C_s^{\text{upper}}  \leq C_{\text{th}})\nonumber \\
&= \mathbb{P}\left(\frac{\gamma_U+1}{ \gamma^*_{E}+1} \leq 2^{2C_{\text{th}}}\right)  \nonumber\\
&=\mathbb{P}(\gamma_U \leq b \gamma^*_{E} + b-1 ) 
\end{align}
respectively, where $a=\pi e 2^{2C_{\text{th}}}/6$ and $b=2^{2C_{\text{th}}}$. Thus, the upper and lower bounds on the SOP can be calculated by appropriately integrating over the PDFs of $\gamma_U$ and $\gamma_E^*$.

Firstly, the upper bound on the SOP can be calculated to yield
\begin{equation} \begin{split} \label{eq:SOP}
P_{\text{SO}}^{\text{upper}} &= \underbrace{\int_{y_1}^{y_4} \int_{\frac{y-3a+\pi e/2}{a}}^{\frac{y_4-3a+\pi e/2}{a}} f_{\gamma_E^*}(x)f_{\gamma_U}(y) \diff x \diff y}_{U_1} + \underbrace{\int_{y_1}^{y_4} \int_{\frac{y_4-3a+\pi e/2}{a}}^{y_4} f_{\gamma_E^*}(x)  f_{\gamma_U}(y) \diff x \diff y}_{U_2}.
\end{split}
\end{equation}
Here, we ignore the $(-3a+\pi e/2)/a$ term in the integration limits, because it is small enough\footnote{The absolute value of this term is less than 3 for $C_{\text{th}}=1$ bit/Hz/s, while $y_1/a$ is larger than $5 \times 10^3$ with the parameters used in Section V.} not to meaningfully affect our calculation. Thus, we calculate the first term $U_1$ to yield (\ref{eq:SOP1}) at the top of the next page. Then, the closed-form expressions for $J_1$, $J_2$, and $J_3$ in (\ref{eq:SOP1}) can be calculated according to (\ref{eq:SOPa}), (\ref{eq:SOPb}), and (\ref{eq:SOPc}).
\begin{figure*}[t]
\label{eq:SOP1}\begin{align}\label{eq:SOP1}  \nonumber
U_1 &\approx  \int_{y_1}^{y_4} \left( F_{\gamma_E^*} \left( \frac{y_4}{a} \right)-F_{\gamma_E^*} \left( \frac{y}{a} \right)  \right) f_{\gamma_U}(y) \diff y = F_{\gamma_E^*} \left( \frac{y_4}{a} \right)\int_{y_1}^{y_4}f_{\gamma_U}(y) \diff y  - \int_{y_1}^{y_4} F_{\gamma_E^*} \left( \frac{y}{a}  \right) f_{\gamma_U}(y) \diff y\nonumber \\
&=F_{\gamma_E^*} \left( \frac{y_4}{a} \right) -  \underbrace{\int_{y_1}^{y_2} e^{\lambda_E \pi \left( Z^2 - \left( \frac{y}{a \zeta}  \right) ^{-\frac{1}{m+3}}\right)} \frac{\left(\frac{y}{\zeta}\right)^{-\frac{1}{m+3}} \left(3  {K_3}\sqrt{\left(\frac{y}{\zeta}\right)^{-\frac{1}{m+3}}-Z^2}+2  {K_4}\right)} {8 {\hat{a}}^2  {\hat{k}} (m+3) y}   \diff y}_{J_1} \nonumber \\
&~~~ -  \underbrace{\int_{y_2}^{y_3} e^{\lambda_E \pi \left( Z^2 - \left( \frac{y}{a \zeta}  \right) ^{-\frac{1}{m+3}}\right)} \frac{\left(\frac{y}{\zeta}\right)^{-\frac{1}{m+3}} \left(3  {K_1}\sqrt{\left(\frac{y}{\zeta}\right)^{-\frac{1}{m+3}}-Z^2}+2  {K_2}\right)} {8 {\hat{a}}^2  {\hat{k}} (m+3) y} \diff y }_{J_2} \nonumber \\
&~~~-\underbrace{\int_{y_3}^{y_4} e^{\lambda_E \pi \left( Z^2 - \left( \frac{y}{a \zeta}  \right) ^{-\frac{1}{m+3}}\right)}  \frac{\pi  \left(\frac{y}{\zeta}\right)^{-\frac{1}{m+3}}}{4  {\hat{a}}^2 {\hat{k}} (m+3) y}  \diff y}_{J_3} 
\end{align}
\hrulefill
\end{figure*}
\begin{figure*}[t]
\begin{subequations}\label{eq:Js}
\begin{align}\label{eq:SOPa}
J_1 &= \frac{1}{8 \hat{a}^2 \hat{k}} \left(  \frac{2 a^{\frac{-1}{m+3}} e^{\lambda_E \pi Z^2}}{\lambda_E \pi} \left(e^{-\lambda_E \pi \left(\frac{y_2}{a \zeta} \right)^{\frac{-1}{m+3}} }- e^{-\lambda_E \pi \left(\frac{y_1}{a \zeta} \right)^{\frac{-1}{m+3}} }\right)K_4 +\frac{3 e^{-\lambda_E \pi Z^2 \left( -1+a^{\frac{1}{m+3}}\right)}K_3}{2 \pi^{3/2}} \right. \nonumber \\
&~ \left. ~~~~\cdot \sum_{i \in \{1,2\} } (-1)^{i+1} \left( \frac{\left(\left( \frac{y_i}{\zeta} \right)^{\frac{-1}{m+3}} -Z^2\right)^{3/2}\left( \sqrt{\pi} - 2 \Gamma \left(\frac{3}{2}, a^{\frac{1}{m+3}}\lambda_E \pi \left(\left( \frac{y_i}{\zeta} \right)^{\frac{-1}{m+3}} -Z^2 \right)    \right)          \right)  }   {\left( a^{\frac{1}{m+3}} \lambda_E \left(\left( \frac{y_i}{\zeta} \right)^{\frac{-1}{m+3}} -Z^2 \right)\right) ^{3/2}   }   \right)  \right) \\  \label{eq:SOPb}
J_2 &= \frac{1}{8 \hat{a}^2 \hat{k}} \left(  \frac{2 a^{\frac{-1}{m+3}} e^{\lambda_E \pi Z^2}}{\lambda_E \pi} \left(e^{-\lambda_E \pi \left(\frac{y_3}{a \zeta} \right)^{\frac{-1}{m+3}} }- e^{-\lambda_E \pi \left(\frac{y_2}{a \zeta} \right)^{\frac{-1}{m+3}} }\right)K_2 +\frac{3 e^{-\lambda_E \pi Z^2 \left( -1+a^{\frac{1}{m+3}}\right)}K_1}{2 \pi^{3/2}} \right. \nonumber \\ 
&~ \left. ~~~~\cdot \sum_{i \in \{2,3\} } (-1)^i \left( \frac{\left(\left( \frac{y_i}{\zeta} \right)^{\frac{-1}{m+3}} -Z^2\right)^{3/2}\left( \sqrt{\pi} - 2 \Gamma \left(\frac{3}{2}, a^{\frac{1}{m+3}}\lambda_E \pi \left(\left( \frac{y_i}{\zeta} \right)^{\frac{-1}{m+3}} -Z^2 \right)    \right)          \right)  }   {\left(a^{\frac{1}{m+3}} \lambda_E \left(\left( \frac{y_i}{\zeta} \right)^{\frac{-1}{m+3}} -Z^2 \right)\right) ^{3/2}   }   \right)  \right) \\  \label{eq:SOPc}
J_3 &= \displaystyle\frac{1}{4 \hat{a}^2 \hat{k} \lambda_E} \left( a^{\frac{-1}{m+3}} e^{\lambda_E \pi Z^2} \left(e^{-\lambda_E \pi \left(\frac{y_4}{a\zeta}  \right)^{\frac{-1}{m+3}}}-e^{-\lambda_E \pi \left(\frac{y_3}{a\zeta}  \right)^{\frac{-1}{m+3}}}         \right)    \right)
\end{align}
\end{subequations}
\hrulefill
\end{figure*}
Finally, the second term  $U_2$ in (\ref{eq:SOP}) can be written as
\begin{equation}
U_2 \approx  F_{\gamma_E^*} \left(y_4 \right)-F_{\gamma_E^*} \left( \frac{y_4}{a} \right).
\end{equation}

The lower bound on the SOP can be calculated to yield
\begin{equation} \begin{split} \label{eq:SOPlower}
P_{\text{SO}}^{\text{lower}} &= {\int_{y_1}^{y_4} \int_{\frac{y-b+1}{b}}^{\frac{y_4-b+1}{b}} f_{\gamma_E^*}(x)f_{\gamma_U}(y) \diff x \diff y} + {\int_{y_1}^{y_4} \int_{\frac{y_4-b+1}{b}}^{y_4} f_{\gamma_E^*}(x)  f_{\gamma_U}(y) \diff x \diff y}.
\end{split}
\end{equation}
Here, we ignore the $(-b+1)/b$ term in a similar manner as in (\ref{eq:SOP}). Therefore, since (\ref{eq:SOPlower}) has an identical expression to (\ref{eq:SOP}), we  have the closed-form of the lower bound on SOP  by  simply changing the variable $a$ to $b$ in the closed-form expression of the upper bound.

From the closed-form expressions for the upper and lower bounds on the SOP, we note that the SOP is inversely proportional to $4 \hat{a}^2 \hat{k}$ (see (\ref{eq:Js})), which is the coverage area of each transmitter. In other words, if the room size is fixed, then the SOP can be decreased by increasing the number of transmitters to reduce the coverage area.

\section{Numerical Results}
\label{sec5}

\begin{table}[t]
\centering
\caption{Simulation Parameters}
\small
\begin{tabular}{  c | c  }
   \Xhline{2\arrayrulewidth}
  \multicolumn{2}{c}{\textbf{Room configuration}} \\  \Xhline{2\arrayrulewidth}
  Length (L) $\times$ Width (W) & $10 \times 10~\mathrm{m^2}$ \\
  Height from the work plane (Z) & $3~\mathrm{m}$ \\
  Number of light fixtures & 4\\
  Number of LEDs per fixture & 9 \\ \Xhline{2\arrayrulewidth}

  \multicolumn{2}{c}{\textbf{LED electrical and optical characteristics}} \\  \Xhline{2\arrayrulewidth}
  Average optical power per LED & $8~\mathrm{W}$\\
  Optical power / current $\eta$ & 5 \\
  Nominal half-intensity angle $\Phi_{1/2}$ & $60^\circ$ \\
  Modulation index $\alpha$ & 0.5 \\  \Xhline{2\arrayrulewidth}

   \multicolumn{2}{c}{\textbf{Optical receiver characteristics}} \\  \Xhline{2\arrayrulewidth}
   Photodetector's responsivity   & 0.54 mA/mW \\
   Photodetector's physical area $A_{PD}$ & $1~\mathrm{cm}^2$ \\
   Lens refractive index $\kappa$ & $1.5$ \\
   Noise power $\sigma^2$ & $-98.35~\mathrm{dBm}$ \\ \hline
\end{tabular}
\label{tb:parameter}
\end{table}

\begin{figure*}[t]
\begin{minipage}[b]{0.5\linewidth}
  \centering
  \centerline{\includegraphics[scale=0.55]{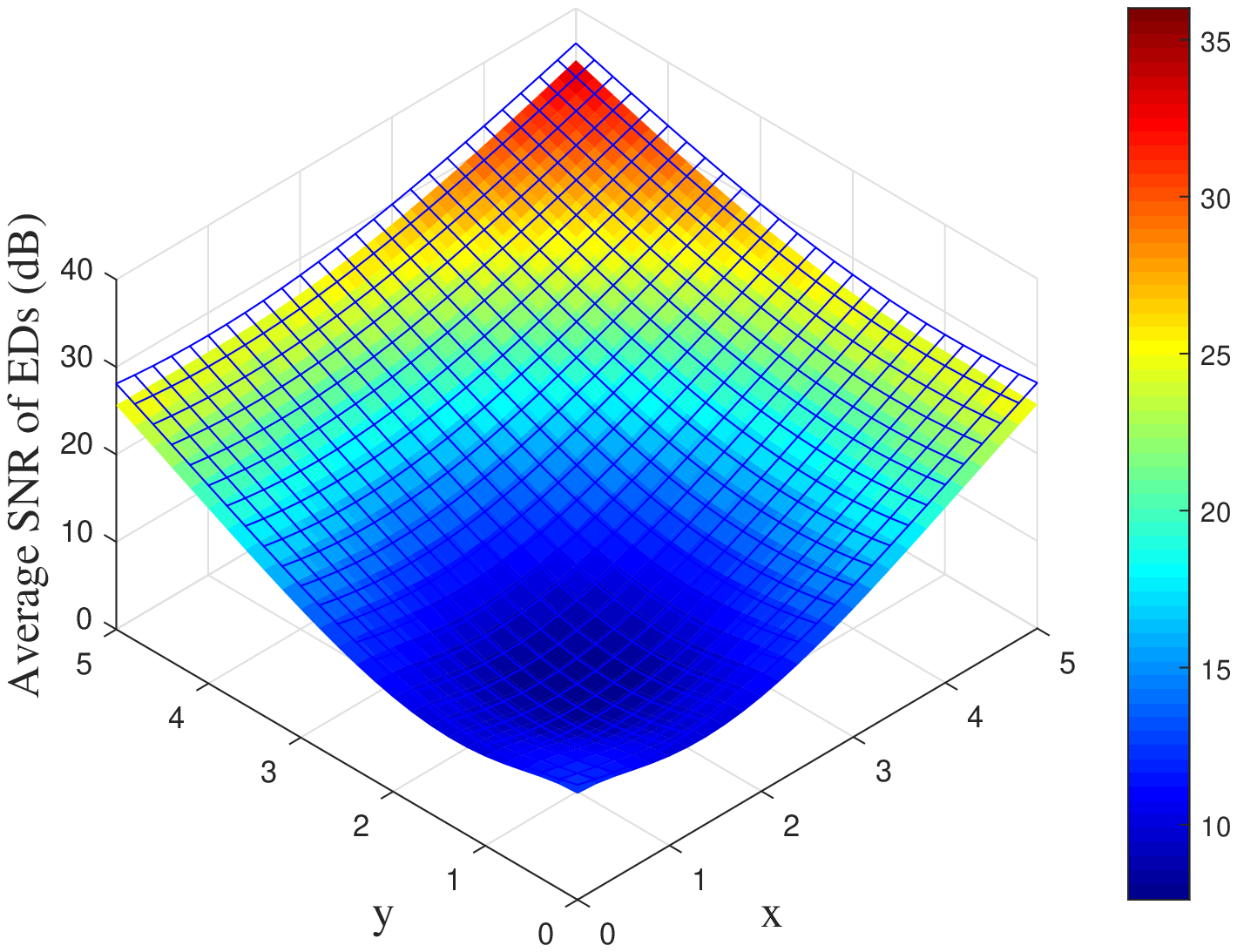}}
 \vspace{0.3cm}
  \centerline{\small (a) The average SNR of EDs}
\end{minipage}
\hfill
\begin{minipage}[b]{0.5\linewidth}
  \centering
  \centerline{\includegraphics[scale=0.55]{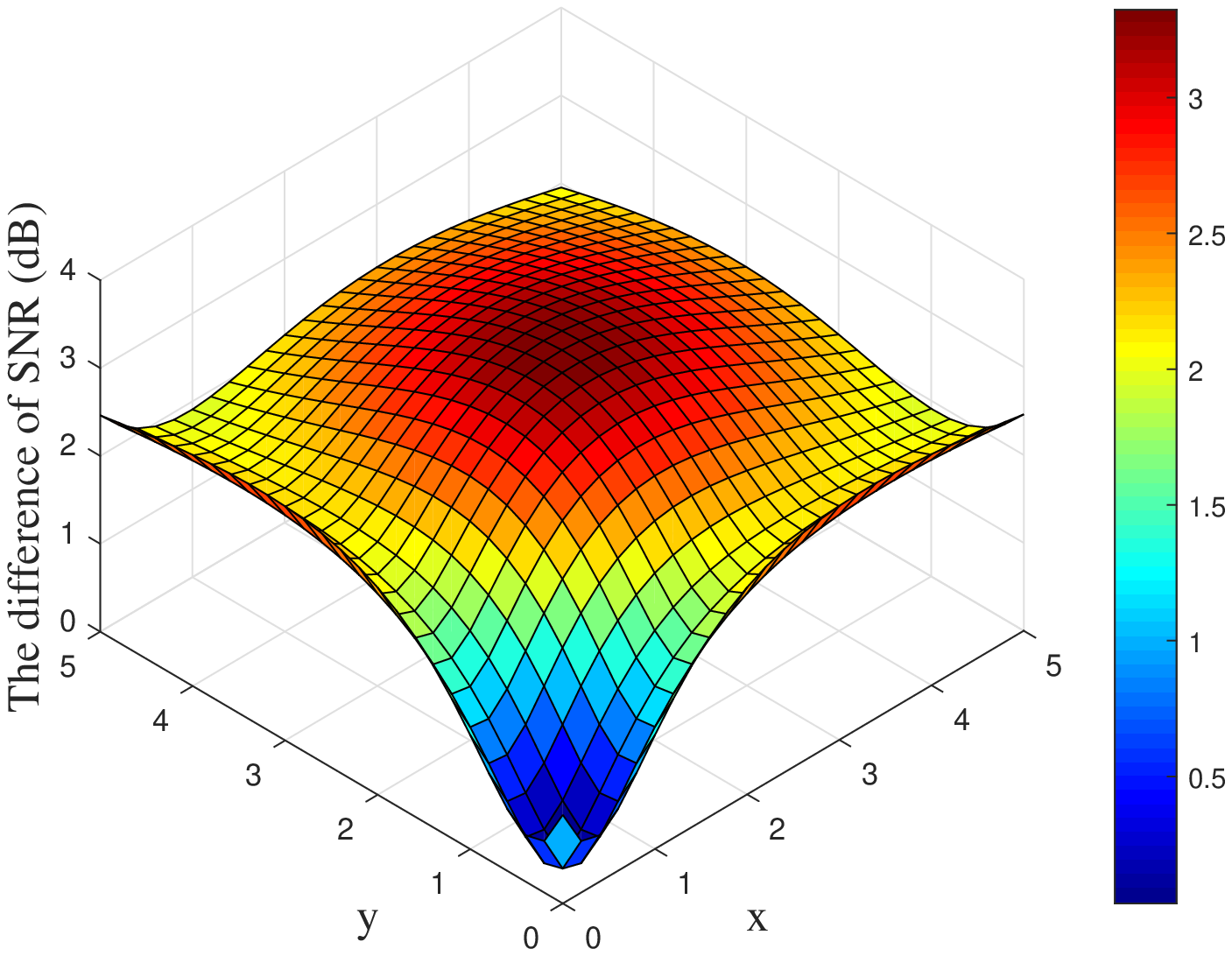}}
  \vspace{0.3cm}
  \centerline{\small (b) The difference in the average SNRs for EDs}
\end{minipage}
 \caption{The average SNR of EDs as a function of the UE location. The bottom surface denotes the result for the optimal beamformer and the top one denotes  LED selection. Four transmitters are located at $(\pm1,\pm1)$. The numerical result for the locations of the UE within the $1$st quadrant are given. The intensity of the ED process is $\lambda_E=0.05$ and the required SNR of the UE is ${\rho}_U=20~\mathrm{dB}$.} \label{fig:Fig3}
\begin{minipage}[b]{0.5\linewidth}
  \centering
  \centerline{\includegraphics[scale=0.55]{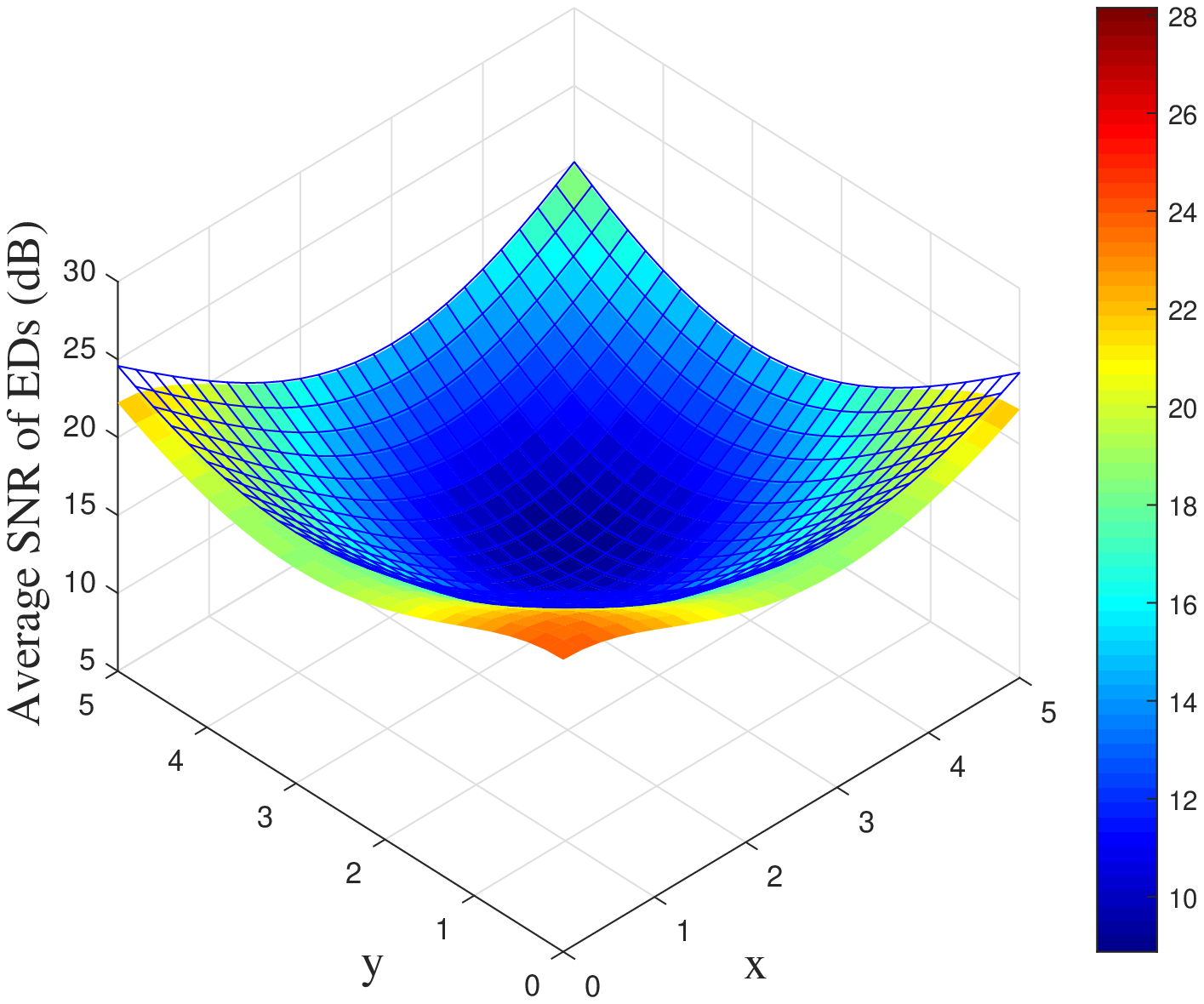}}
 \vspace{0.3cm}
  \centerline{\small (a) The average SNR of EDs}
\end{minipage}
\hfill
\begin{minipage}[b]{0.5\linewidth}
  \centering
  \centerline{\includegraphics[scale=0.55]{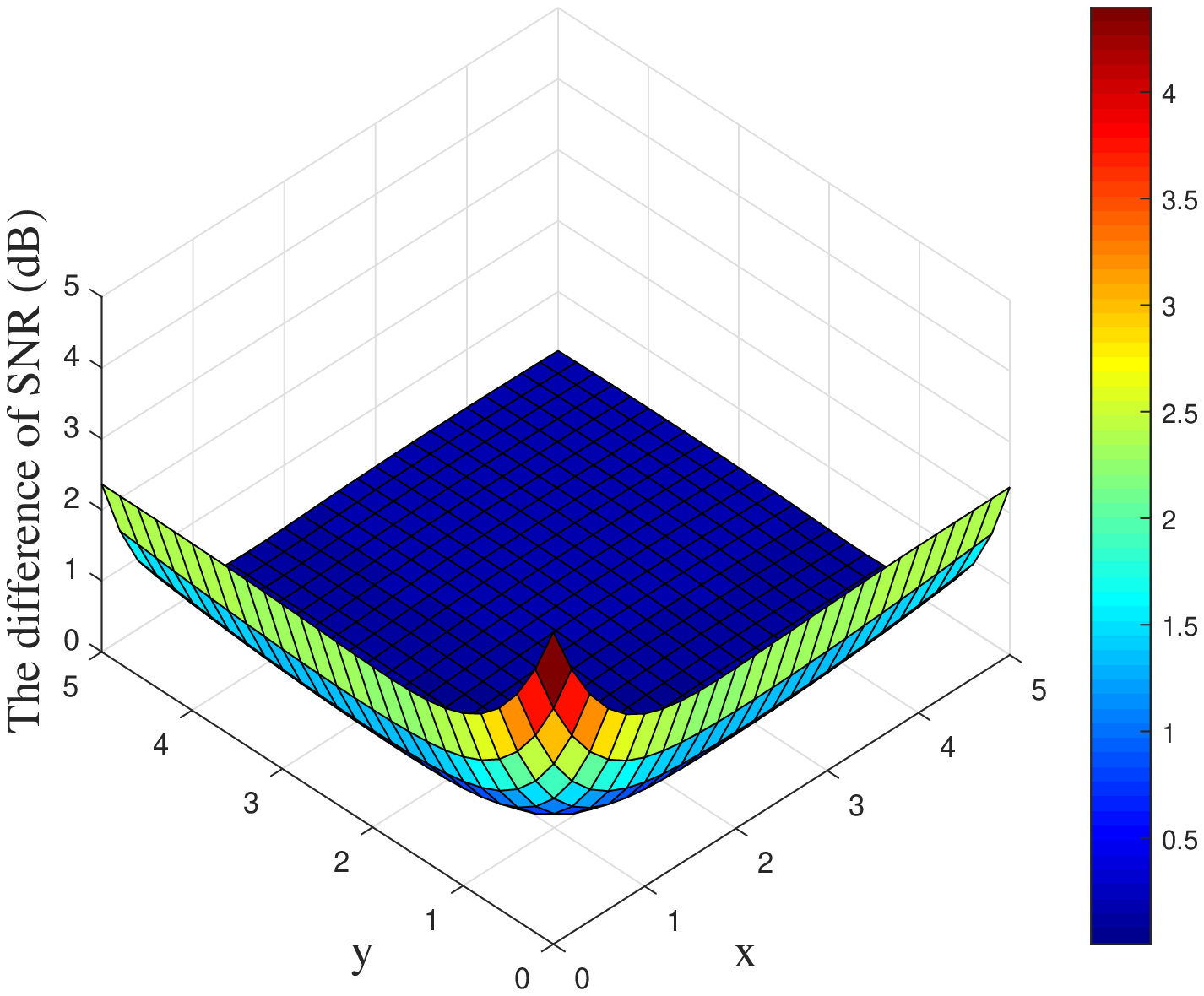}}
  \vspace{0.3cm}
  \centerline{\small (b) The difference in the average SNRs for EDs}
\end{minipage}
 \caption{The average SNR of EDs as a function of the UE location. The bottom surface denotes the result for the optimal beamformer and the top one denotes  LED selection. Four transmitters are located at $(\pm3,\pm3)$. The numerical result for the locations of the UE within the $1$st quadrant are given. The intensity of the ED process is $\lambda_E=0.05$ and the required SNR of the UE is ${\rho}_U=20~\mathrm{dB}$.} \label{fig:Fig4}
\end{figure*}

\begin{figure*}[t]
\begin{minipage}[b]{0.5\linewidth}
  \centering
  \centerline{\includegraphics[scale=0.55]{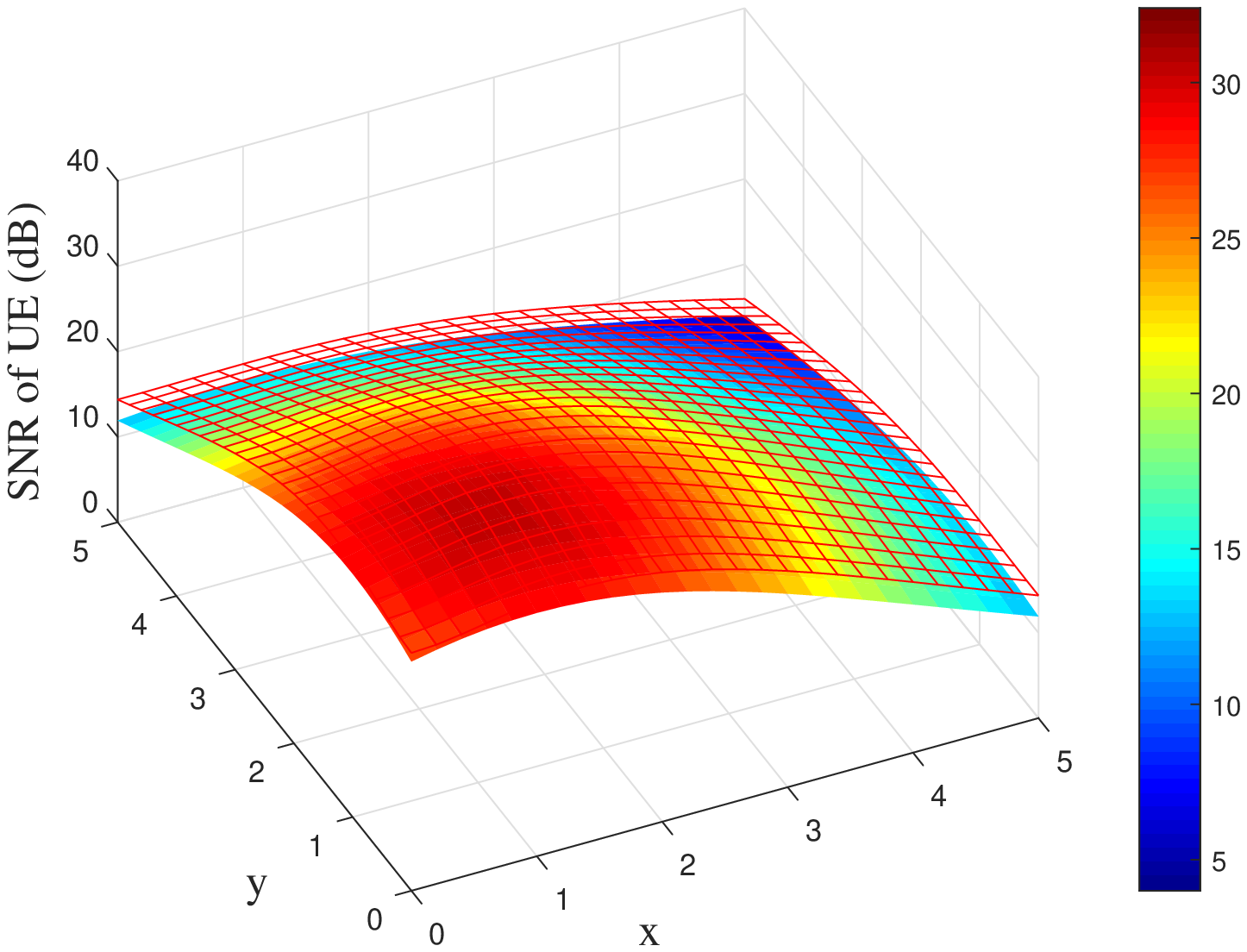}}
 \vspace{0.3cm}
  \centerline{\small (a) The SNR of the UE}
\end{minipage}
\hfill
\begin{minipage}[b]{0.5\linewidth}
  \centering
  \centerline{\includegraphics[scale=0.55]{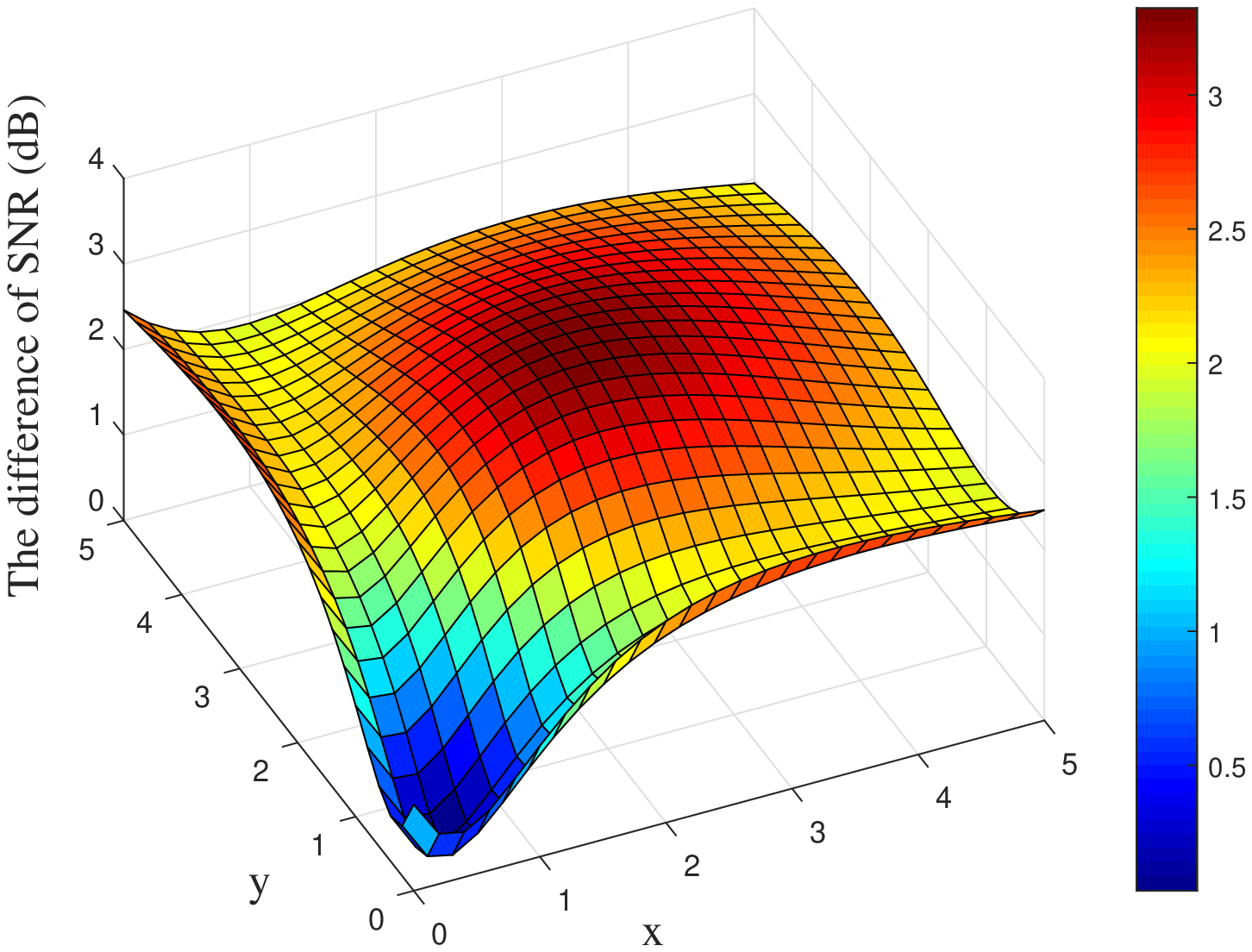}}
  \vspace{0.3cm}
  \centerline{\small (b) The difference in the SNRs for the UE}
\end{minipage}
 \caption{The average SNR of the UE as a function of the UE locations. The top surface denotes the result for the optimal beamformer and the bottom one denotes  LED selection. Four transmitters are located at $(\pm1,\pm1)$. The numerical results for the locations of the UE within the $1$st quadrant are given. The density of the ED process is $\lambda_E=0.05$ and the constraint on the average SNR of EDs is $\overline{\rho}_E=20~\mathrm{dB}$.} \label{fig:Fig5}
\begin{minipage}[b]{0.5\linewidth}
  \centering
  \centerline{\includegraphics[scale=0.55]{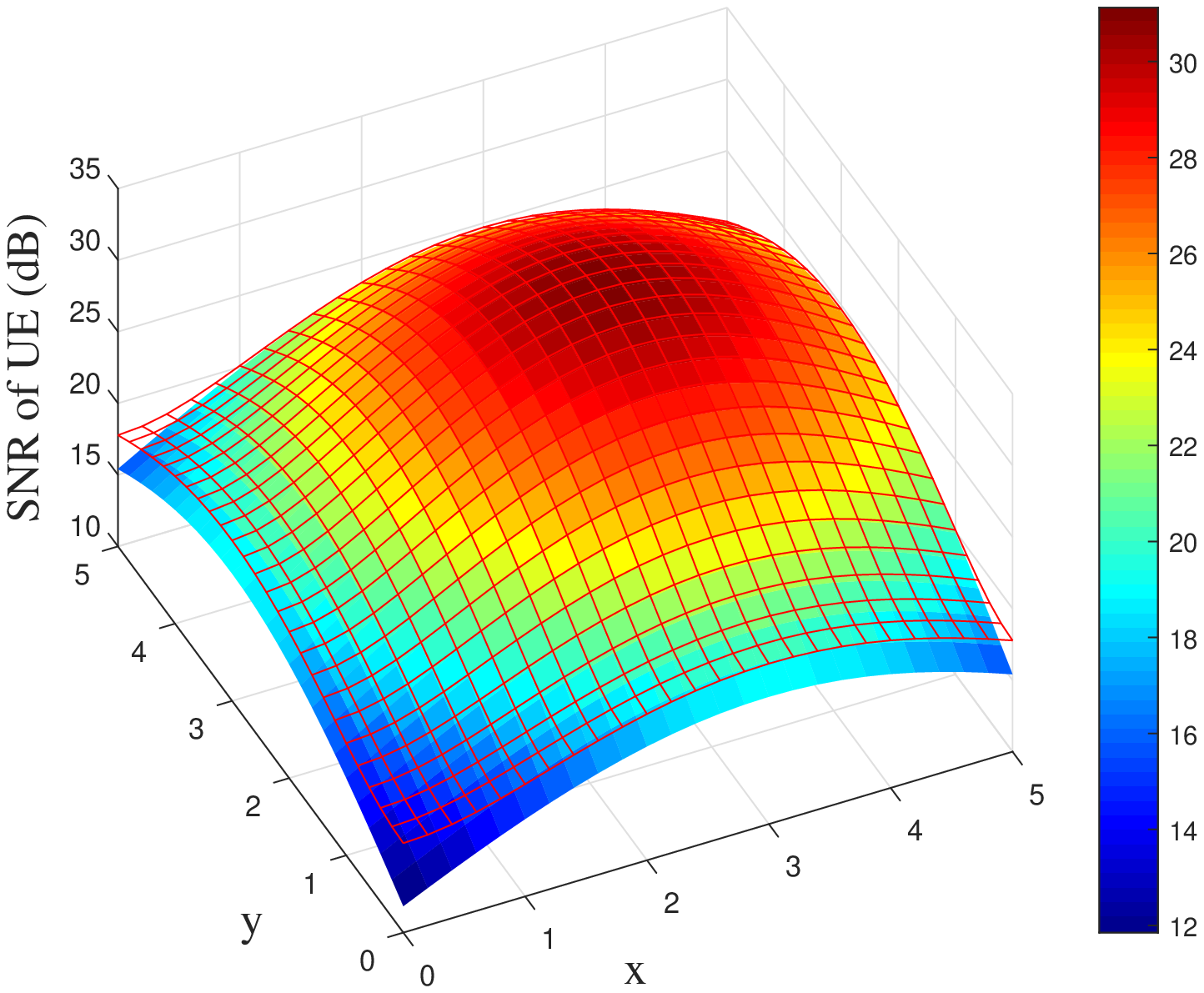}}
 \vspace{0.3cm}
  \centerline{\small (a) The SNR of the UE}
\end{minipage}
\hfill
\begin{minipage}[b]{0.5\linewidth}
  \centering
  \centerline{\includegraphics[scale=0.55]{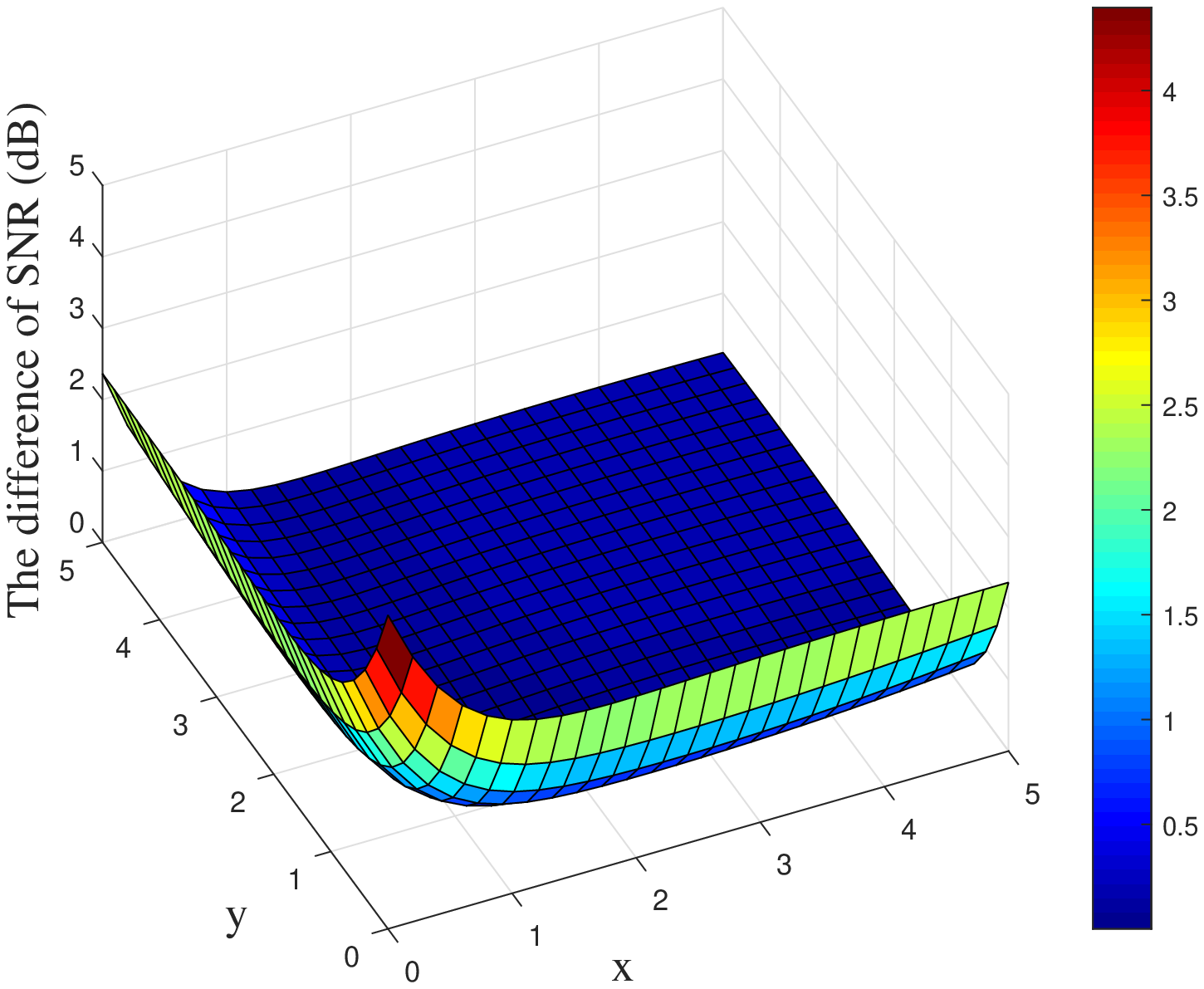}}
  \vspace{0.3cm}
  \centerline{\small (b) The difference in the SNRs for UE}
\end{minipage}
 \caption{The average SNR of the UE as a function of the UE location. The top surface denotes the result for the optimal beamformer and the bottom one denotes  LED selection. Four transmitters are located at $(\pm3,\pm3)$. The numerical results for the locations of the UE within the $1$st quadrant are given. The intensity of the ED process is $\lambda_E=0.05$ and the constraint on the average SNR of EDs  is $\overline{\rho}_E=20~\mathrm{dB}$.} \label{fig:Fig6}
\end{figure*}

In this section, we provide  numerical results to verify our analysis. The room configuration and simulation parameters are provided in Table~\ref{tb:parameter}. We use the Cartesian coordinate system to identify positions of transmitters and receivers, where the center of the room is located at the origin. We consider a VLC network where four transmitters are symmetrically located, i.e., their locations can be described by $(\pm d_0, \pm d_0)$. A homogeneous PPP describes the ED locations.

\subsection{Signal-to-Noise Ratio}
Figs.~\ref{fig:Fig3} and~\ref{fig:Fig4} show the comparison of the average SNR of EDs by using the beamforming and LED selection schemes for the different locations of transmitters, where $\lambda_E=0.05$ and $\rho_U=20~\mathrm{dB}$. Since four transmitters are symmetrically located, we only show the results for the UE locations within the first quadrant. The bottom and top surfaces denote the average SNR of an ED as a function of the UE location when using the beamforming and LED selection schemes, respectively. Firstly, when four transmitters are located closely to each other in the center area of the room as in Fig.~\ref{fig:Fig3}, the transmitters emit a signal with high power to the UE, which is located in the outer area to satisfy $\rho_U$. Due to the broadcasting characteristic of  light, the EDs can eavesdrop the signal easily; thus their average SNR also increases as shown in Fig.~\ref{fig:Fig3}(a). In this case, the EDs can achieve a higher average SNR than $\rho_U$ if the UE is away from the transmitter. On the contrary, when the distance between the transmitter and the UE decreases, one can see that the average SNR of the EDs also decreases because of the decrease in transmitted power. Furthermore, in Fig.~\ref{fig:Fig3}(b), it is shown that the difference in the average SNR of EDs between  optimal beamforming and LED selection is small when the UE is located near to the transmitter. However,  beamforming can outperform  LED selection when the UE is located in the outer area. On the other hand, we  see that when the four transmitters are adequately separated from each other as in Fig.~\ref{fig:Fig4}, both  beamforming and  LED selection  exploit the spatial advantage to decrease the average SNR of EDs. The average SNR of EDs is less than $\rho_U$ for most of the area except for around the center point of the transmitters. 

Figs.~\ref{fig:Fig5} and~\ref{fig:Fig6} show a comparison of the SNR of the UE using  beamforming and LED selection schemes for  different locations of transmitters, where $\lambda_E=0.05$ and $\overline{\rho}_E=20~\mathrm{dB}$. The top and bottom surfaces denote the SNR of the UE as a function of the UE location when using  beamforming and LED selection, respectively. Firstly, when the transmitters are located closely together as in Fig.~\ref{fig:Fig5}, the SNR of the UE can be high when the UE is located near to the transmitter as shown in Fig.~\ref{fig:Fig5}(a). However, when the UE moves away from the transmitter, it is difficult for the UE to achieve an SNR higher than  $\overline{\rho}_E$. This is because the transmitter should convey the signal with high enough power to reach the UE, which enables the EDs to overhear the signal easily. In addition, we see from Fig.~\ref{fig:Fig5}(b) that the SNR difference for beamforming and  selection is not significant for the UE  located near the transmitter, but it increases as the distance between the transmitter and the UE grows. In contrast, when the transmitters are sufficiently separated as in Fig.~\ref{fig:Fig6}, we  note that the UE can achieve a higher SNR than $\overline{\rho}_E$ over almost the entire area of the room except for at the center point of the room as shown in Fig.~\ref{fig:Fig6}(a). This can be possible by selectively transmitting a signal to the UE without excessively increasing the signal power of other transmitters, thus exploiting the spatial benefit. Similarly, from Fig.~\ref{fig:Fig6}(b), we  note that the SNR difference for the UE is not significant when the UE is close to the transmitter.

These results show that optimal beamforming effectively transmits a signal to the UE trying not to expose the signal to EDs when the transmitters are adequately separated. Moreover, it is shown that  beamforming has better performance than  LED selection in all cases. However, when the transmitters are adequately separated, and the UE is near to the transmitters, it is shown that the difference in  secrecy performance between the two schemes is small. Therefore, considering the computational complexity and feasibility of  optimal beamforming, we conclude that LED selection may be an attractive option in some scenarios.

\begin{figure}[t]
\centering
\includegraphics[scale=0.6]{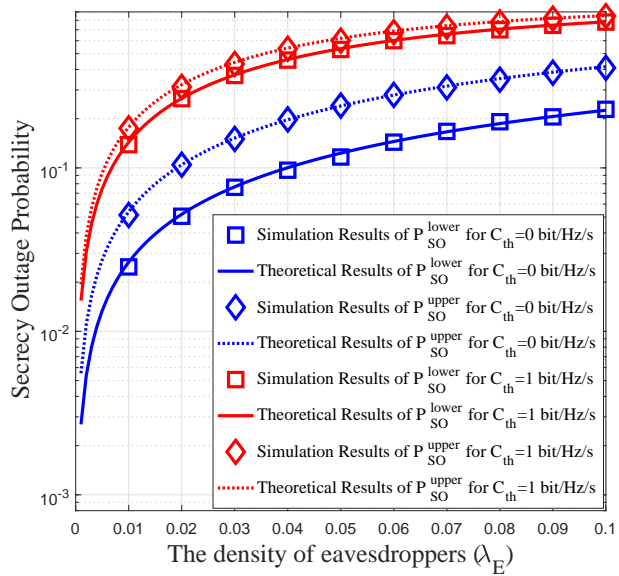}
\caption{The upper and lower bounds on the SOP with  LED selection  for  different $C_{\text{th}}$ according to different intensities of EDs $\lambda_E$, where $N=4 \times 4$, $g=1~\mathrm{m}$, $\hat{a}=1~\mathrm{m}$, $\hat{k}=1.25$, $L=10~\mathrm{m}$, and $W=12 ~\mathrm{m}$.}
\label{fig:Fig7}
\end{figure}

\begin{figure}[t]
\centering
\includegraphics[scale=0.6]{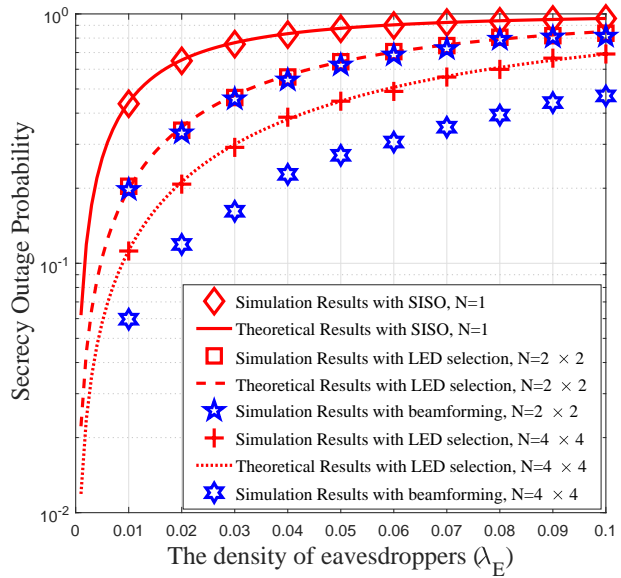}
\caption{A comparison of the upper bound on the SOP between  LED selection and beamforming for  different numbers of transmitters according to different intensities of EDs $\lambda_E$, where $C_{\text{th}}=0.5$ bit/Hz/s, $g=1~\mathrm{m}$, $L=10~\mathrm{m}$, and $W=12~\mathrm{m}$.}
\label{fig:Fig8}
\end{figure}

\subsection{Secrecy Outage Probability}
Fig.~\ref{fig:Fig7} shows the upper and lower bounds on the SOP with  LED selection for different intensities of EDs and different values of $C_{\text{th}}$, where $N=4 \times 4$, $g=1~\mathrm{m}$, $\hat{a}=1~\mathrm{m}$, $\hat{k}=1.25$, $L=10~\mathrm{m}$, and $W=12~\mathrm{m}$. Both  simulated results and theoretical results are presented, which are shown to perfectly match. As can be seen from the figure, both the SOP upper and lower bounds increase as $\lambda_E$ and $C_{\text{th}}$ increase, as expected.

Fig.~\ref{fig:Fig8} shows a comparison of the upper bound on  SOP between  beamforming and LED selection for  different numbers of transmitters and different values of $\lambda_E$, where $C_{\text{th}}=0.5$ bit/Hz/s, $g=1~\mathrm{m}$, $L=10~\mathrm{m}$, and $W=12~\mathrm{m}$. It can be seen that as the number of transmitters increases, the upper bound on SOP decreases. Furthermore, when the number of transmitters is small, i.e., $N= 2\times2$, the difference in the SOP for  beamforming and LED selection is small. However, when $N$ is large, we can see that the upper bound on the SOP with beamforming is less than for LED selection. Since a large number of transmitters can exploit the excessive spatial degrees of freedom to steer the signal toward the UE, the transmitters can significantly increase the SNR of the UE while suppressing the signal everywhere else inside the room. However, we need to consider that finding the maximum eigenvalue and its associated eigenvector with a large number of transmitters requires high computation complexity, which increases proportionally to $N^2$. 

\section{Conclusion}
In this paper, we studied  optimal MISO beamforming schemes and  a suboptimal LED selection scheme to enhance the secrecy performance in VLC systems when multiple EDs are randomly distributed throughout the communication region. By using the MISO beamforming scheme, we can minimize the average SNR of EDs (or indeed the worst case SNR of EDs) and maximize the SNR of the UE with only statistical information about ED locations. The LED selection scheme is not superior to the optimal beamformer in the respect of secrecy performance; however, when the UE is located near to one of the transmitters, LED selection provides a good practical solution to enhancing secrecy performance without high computational complexity. Based on LED selection, closed-form approximations for the upper and lower bounds on the SOP were derived. Our results provide useful insight and analytic tools that can be used to enhance the secrecy in VLC systems and give a solid basis for further study.

\appendices
\section{Proof of the Lemma 1}
\label{sec:app_A}

\begin{figure}[t]
  \centering
  \centerline{\includegraphics[scale=0.5]{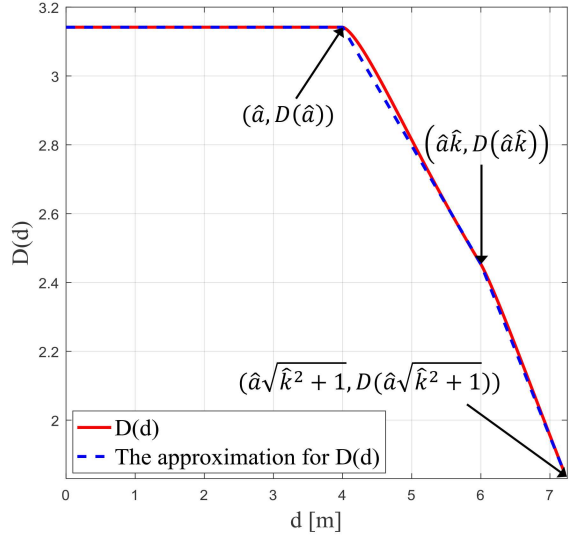}}
 \caption{\small An example of $\hat{D}(d)$ for $\hat{a} = 4~\mathrm{m}$, $\hat{k}=1.5$. Note $\hat{a}\hat{k}=6~\mathrm{m}$ and $\hat{a}\sqrt{\hat{k}^2+1}=7.21~\mathrm{m}$. $\hat{a}\hat{k}=6~\mathrm{m}$ and $\hat{a}\sqrt{\hat{k}^2+1}=7.21~\mathrm{m}$} \label{fig:piecewise}
\end{figure}

The UE is randomly located according to a homogeneous BPP in the shaded area. Therefore, the distance in the work plane between the transmitter and the UE $d_U$ cannot exceed $\hat{a}\sqrt{\hat{k}^2+1}$. Thus, the CDF of $d_U$ is given by
\begin{equation}
\begin{array}{lll}
F_{d_{U}}(d)=\displaystyle\frac{A(d)}{4 \hat{k}\hat{a}^2} &\mbox{for} &0 < d \leq \hat{a}\sqrt{\hat{k}^2+1}
\end{array}
\label{eq:CDF_of_d_U}
\end{equation}
where $A(d)$ denotes the area of the circle bounded by the rectangle as shown in Fig.~\ref{fig:SOP_configuration}. $A(d)$ is described as (\ref{eq:area}) at the top of the this page.
\begin{figure*}[t]
\begin{equation}\begin{split}
A(d)=
 \left \{ \begin{array}{l l l}
\displaystyle{\pi d^2}    ~~~~~~~~~~~~~~~~~~~~~~~~~~~~~~~~~~~~~~~~~~~~~~~~~~~~~~~~~~~~~~~~~~~~ \mbox{for}  ~~ 0 < d \leq \hat{a} \\ \\
\displaystyle{\pi d^2 - 2 \left( d^2 \arccos \left( \frac{\hat{a}}{d} \right) - \hat{a} \sqrt{d^2 - \hat{a}^2} \right)}   ~~~~~~~~~~~~~~~~~~~~~~~~~~~~~  \mbox{for}   ~~                         \hat{a} < d \leq \hat{k}\hat{a}  \\  \\
\displaystyle{\pi d^2 - 2 \left( d^2 \arccos \left( \frac{\hat{a}}{d} \right) - \hat{a} \sqrt{d^2 - \hat{a}^2} \right) }                  {-2 \left( d^2 \arccos \left( \frac{\hat{a}\hat{k}}{d} \right) - \hat{a}\hat{k} \sqrt{d^2 - (\hat{a}\hat{k})^2} \right) } \\ ~~~~~~~~~~~~~~~~~~~~~~~~~~~~~~~~~~~~~~~~~~~~~~~~~~~~~~~~~~~~~~~~~~~~~~~~~ \mbox{for}  ~~  \hat{k}\hat{a} < d \leq \hat{a}\sqrt{\hat{k}^2+1}    \\
\end{array} \right.
\label{eq:area}
\end{split}
\end{equation}
\hrulefill
\end{figure*}

Thus, $A(d)$ can be described as  $A(d)=D(d) \cdot d^2$, where $D(d)$ can be approximated by applying a piecewise approximation with a linear function of $d$, i.e.,
\begin{equation}\begin{split}
\hat{D}(d)=\left \{ \begin{array}{lll}
\pi &\mbox{for}& 0 < d \leq \hat{a} \\
K_1d+K_2 &\mbox{for} &\hat{a} < d \leq \hat{k}\hat{a}  \\
K_3d+K_4 &\mbox{for} &\hat{k}\hat{a} < d \leq \hat{a}\sqrt{\hat{k}^2+1}. \end{array}\right.
\end{split}\end{equation}
To find the optimal $K_i$ for $i \in \{1,2,3,4\}$, we evaluate three coordinates $D(\hat{a})$, $D(\hat{a}\hat{k})$, and $D(\hat{a}\sqrt{\hat{k}^2+1})$ as shown in Fig.~\ref{fig:piecewise}. Using these values, we can easily calculate the approximation constants as shown in~(\ref{eq:K}).

Finally, from (\ref{simplified_h}) and (\ref{eq:SNR1}), the SNR of the UE can be described as a function of $d$ according to
\begin{equation}
\label{eq:gamma}
\gamma_{U}(d)=\displaystyle \frac{\alpha^2 I_{DC}^2 K^2 (d^2 + Z^2)^{-(m+3)}}{\sigma^2} = \zeta (d^2+Z^2)^{-(m+3)}
\end{equation}
where $\zeta=(\alpha^2 I_{DC}^2 K^2)/\sigma^2$. Thus, the CDF and PDF of $\gamma_U$ can be written as (\ref{eq:CDF_SNR_U}) and (\ref{eq:PDF_SNR_U}), respectively.

\section{Proof of the Lemma 2}
\label{sec:app_B}
Since the EDs are randomly distributed according to a homogeneous PPP $\Phi_{E}$ with intensity $\lambda_E$ on the work plane, the PDF of the minimum distance in the work plane between the selected transmitter and the nearest ED, i.e., $d_{E}^* = \underset{e \in \Phi_E}{\min}~{d_{E_e}}$, where $d_{E_e}$ is the distance in the work plane between the selected transmitter and $E_e$, is given by
\begin{equation}
F_{d_{E}^*}(d)=1 - \exp(-\lambda_E \pi d^2)
\end{equation}
for $0 \leq d \leq \infty$.  This follows from  contact distance distribution to the nearest point of the PPP~\cite{baddeley2015spatial}. Here, the unbounded upper limit implies that the number of EDs can be zero. The PDF of $d_{E}^*$ can be calculated as $f_{d_{E}^*}(d)=2 \lambda_E \pi d \exp(-\lambda_E \pi d^2)$ for $0 \leq d \leq \infty$. Therefore, since $\gamma_E^*$ also has the same relation with $d_{E}^*$ as~(\ref{eq:gamma}), the PDF and CDF of $\gamma_E^*$ can be calculated as below.
\begin{align}
f_{\gamma_E^*}(x)&=\frac{\lambda_E \pi \left(\frac{x}{\zeta} \right)^{-\frac{1}{m+3}}}{x(m+3)} e^{-\lambda_E \pi \left( -Z^2+\left(\frac{x}{\zeta} \right)^{-\frac{1}{m+3}}\right)}, \\
F_{\gamma_E^*}(x)&= \int_0^x f_{\gamma_E^*}(u) \diff u \nonumber \\ \nonumber
&= \int_0^x \frac{\lambda_E \pi \left(\frac{u}{\zeta} \right)^{-\frac{1}{m+3}}}{u(m+3)} e^{-\lambda_E \pi \left( -Z^2+\left(\frac{u}{\zeta} \right)^{-\frac{1}{m+3}}\right)} \diff u \\ \nonumber
&=\int_{\left(\frac{x}{\zeta} \right)^{-\frac{1}{m+3}}}^\infty \lambda_E \pi e^{-\lambda_E \pi (v-Z^2)} \diff v \\
&=e^{\lambda_E \pi \left( Z^2 -\left(\frac{x}{\zeta} \right)^{-\frac{1}{m+3}}\right)}
\end{align}
for $0 \leq x \leq \zeta Z^{-2(m+3)}$, respectively, where $v=\left( \frac{u}{\zeta}\right)^{-\frac{1}{m+3}}$.


\ifCLASSOPTIONcaptionsoff
  \newpage
\fi

\bibliographystyle{IEEEtran}
\bibliography{reference_sunghwan}

\begin{thebibliography}{10}
\providecommand{\url}[1]{#1}
\csname url@samestyle\endcsname
\providecommand{\newblock}{\relax}
\providecommand{\bibinfo}[2]{#2}
\providecommand{\BIBentrySTDinterwordspacing}{\spaceskip=0pt\relax}
\providecommand{\BIBentryALTinterwordstretchfactor}{4}
\providecommand{\BIBentryALTinterwordspacing}{\spaceskip=\fontdimen2\font plus
\BIBentryALTinterwordstretchfactor\fontdimen3\font minus
  \fontdimen4\font\relax}
\providecommand{\BIBforeignlanguage}[2]{{%
\expandafter\ifx\csname l@#1\endcsname\relax
\typeout{** WARNING: IEEEtran.bst: No hyphenation pattern has been}%
\typeout{** loaded for the language `#1'. Using the pattern for}%
\typeout{** the default language instead.}%
\else
\language=\csname l@#1\endcsname
\fi
#2}}
\providecommand{\BIBdecl}{\relax}
\BIBdecl

\bibitem{lifi}
H.~Haas, L.~Yin, Y.~Wang, and C.~Chen, ``What is \textsc{L}i\textsc{F}i?''
  \emph{Journal of Lightwave Technology}, vol.~34, no.~6, pp. 1533--1544, March
  2016.

\bibitem{5gwillbe}
J.~G. Andrews, S.~Buzzi, W.~Choi, S.~V. Hanly, A.~Lozano, A.~C.~K. Soong, and
  J.~C. Zhang, ``What will \textsc{5G} be?'' \emph{IEEE Journal on Selected
  Areas in Commun.}, vol.~32, no.~6, pp. 1065--1082, June 2014.

\bibitem{A.D75}
A.~D. Wyner, ``The wire-tap channel,'' \emph{Bell Syst. Tech. J.}, vol.~54, pp.
  1355--1387, Jan. 1975.

\bibitem{M.H08}
M.~Haenggi, ``The secrecy graph and some of its properties,'' in \emph{Proc.
  IEEE Int. Symp. Inf. Theory in Toronto, Canada}, July 2008.

\bibitem{P.C08}
P.~C. Pinto, J.~Barros, and M.~Z. Win, ``Physical-layer security in stochastic
  wireless networks,'' in \emph{Proc. IEEE Int. Conf. Commun. Syst. in
  Guangzhou, China}, Nov. 2008.

\bibitem{X.Z11}
X.~Zhou, R.~K. Ganti, J.~G. Andrews, and A.~Hjorungnes, ``On the throughput
  cost of physical layer security in decentralized wireless networks,''
  \emph{IEEE Trans. on Wireless Commun.}, vol.~10, no.~8, pp. 2764--2775, Aug.
  2011.

\bibitem{G.G14}
G.~Geraci, S.~Singh, J.~G. Andrews, J.~Yuan, and I.~B. Collings, ``Secrecy
  rates in broadcast channels with confidential messages and external
  eavesdroppers,'' \emph{IEEE Trans. on Wireless Commun.}, vol.~13, no.~5, pp.
  2931--2943, May 2014.

\bibitem{T.X14}
T.~X. Zheng, H.~M. Wang, and Q.~Yin, ``On transmission secrecy outage of a
  multi-antenna system with randomly located eavesdroppers,'' \emph{IEEE
  Commun. Lett.}, vol.~18, no.~8, pp. 1299--1302, Aug. 2014.

\bibitem{G.C17}
G.~Chen, J.~P. Coon, and M.~D. Renzo, ``Secrecy outage analysis for downlink
  transmissions in the presence of randomly located eavesdroppers,'' \emph{IEEE
  Trans. Inform. Forensics and Security, to appear}, Apr. 2017.

\bibitem{lampeJSAC}
A.~Mostafa and L.~Lampe, ``Physical-layer security for \textsc{MISO} visible
  light communication channels,'' \emph{IEEE Journal on Selected Areas in
  Commun.}, vol.~33, no.~9, pp. 1806--1818, Sep. 2015.

\bibitem{lampeICC}
------, ``Physical-layer security for indoor visible light communications,'' in
  \emph{IEEE ICC in Sydney, Australia}, June 2014.

\bibitem{lampeGlobecom}
------, ``Securing visible light communications via friendly jamming,'' in
  \emph{IEEE Globecom Workshops in Austin, USA}, Dec. 2014.

\bibitem{alouini1}
H.~Zaid, Z.~Rezki, A.~Chaaban, and M.~S. Alouini, ``Improved achievable secrecy
  rate of visible light communication with cooperative jamming,'' in \emph{IEEE
  GlobalSIP in Orlando, USA}, Dec. 2015.

\bibitem{alouini2}
M.~A. Arfaoui, Z.~Rezki, A.~Ghrayeb, and M.~S. Alouini, ``On the secrecy
  capacity of \textsc{MISO} visible light communication channels,'' in
  \emph{IEEE Globecom in Washington D.C., USA}, Dec. 2016.

\bibitem{17S.C}
S.~Cho, G.~Chen, and J.~P. Coon, ``Secrecy analysis in visible light
  communication systems with randomly located eavesdroppers,'' in \emph{2017
  IEEE ICC Workshops in Paris, France}, May 2017.

\bibitem{TKomine}
T.~Komine and M.~Nakagawa, ``Fundamental analysis for visible-light
  communication system using \textsc{LED} lights,'' \emph{IEEE Consumer
  Electronics}, vol.~50, no.~1, pp. 100--107, Feb. 2004.

\bibitem{optic_capacity}
A.~Lapidoth, S.~M. Moser, and M.~A. Wigger, ``On the capacity of free-space
  optical intensity channels,'' \emph{IEEE Trans. on Inform. Theory}, vol.~55,
  no.~10, pp. 4449--4461, Oct. 2009.

\bibitem{capacity_closedform}
O.~Ozel, E.~Ekrem, and S.~Ulukus, ``Gaussian wiretap channel with an amplitude
  constraint,'' in \emph{IEEE Inform. Theory Workshop in Lausanne,
  Switzerland}, Sep. 2012.

\bibitem{09I.K}
I.~Krikidis, J.~S. Thompson, and S.~Mclaughlin, ``Relay selection for secure
  cooperative networks with jamming,'' \emph{IEEE Trans. on Wireless Commun.},
  vol.~8, no.~10, pp. 5003--5011, October 2009.

\bibitem{15H.H}
H.~Hui, A.~L. Swindlehurst, G.~Li, and J.~Liang, ``Secure relay and jammer
  selection for physical layer security,'' \emph{IEEE Signal Proc. Lett.},
  vol.~22, no.~8, pp. 1147--1151, Aug 2015.

\bibitem{16G.C}
G.~Chen, Y.~Gong, P.~Xiao, and J.~A. Chambers, ``Dual antenna selection in
  secure cognitive radio networks,'' \emph{IEEE Trans. on Vehicular
  Technology}, vol.~65, no.~10, pp. 7993--8002, Oct 2016.

\bibitem{04S.B}
S.~Boyd and L.~Vandenberghe, \emph{Convex Optimization}.\hskip 1em plus 0.5em
  minus 0.4em\relax New York, NY, USA: Cambridge University Press, 2004.

\bibitem{baddeley2015spatial}
A.~Baddeley, E.~Rubak, and R.~Turner, \emph{Spatial point patterns: methodology
  and applications with \textsc{R}}.\hskip 1em plus 0.5em minus 0.4em\relax CRC
  Press, 2015.

\end{thebibliography}

\end{document}